\documentclass[twoside,leqno]{article}

\usepackage[letterpaper]{geometry}

\usepackage{ltexpprt}
\usepackage{hyperref}

\usepackage{graphicx}

\usepackage[bbgreekl]{mathbbol}
\DeclareSymbolFontAlphabet{\mathbbl}{bbold}

\usepackage{enumerate}
\usepackage{xcolor}

\usepackage{enumitem}

\usepackage{algorithm}
\usepackage{algpseudocode}

\usepackage{amsmath}
\usepackage{amssymb}
\usepackage{thmtools, thm-restate}
\usepackage{ntheorem,cleveref}

\newtheorem{definition}{Definition}[section]
\newtheorem{remark}{Remark}[section]

\newtheorem{claim}{Claim}[section]
\newtheorem{algo}{Algorithm}[section]
\newtheorem{invariant}{Invariant}[section]

\renewtheorem{lemma}{Lemma}[section]
\renewtheorem{corollary}{Corollary}[section]
\renewtheorem{proposition}{Proposition}[section]
\renewtheorem{theorem}{Theorem}[section]

\newcommand{\ceil}[1]{\left\lceil #1 \right\rceil}

\newcommand{\union}{\cup}
\newcommand{\inter}{\cap}

\newcommand{\card}[1]{\left|#1\right|}

\newcommand{\C}{\mathcal{C}}

\newcommand{\R}{\mathbb{R}}

\newcommand{\vol}{\mathrm{Vol}}
\newcommand{\polylog}{\mathrm{polylog}\,}
\newcommand{\poly}{\mathrm{poly}}

\newcommand{\boundary}{\partial}

\DeclareMathOperator*{\argmin}{arg\,min}
\newcommand{\PAR}[1]{\left( #1 \right)}
\newcommand{\lmin}{\lambda_{\min}}
\newcommand{\lmax}{\lambda_{\max}}
\newcommand{\eps}{\epsilon}

\newcommand{\var}[1]{\texttt{#1}}

\newcommand\flag[1]{%
    \leavevmode\marginpar{%
        \raisebox{\dimexpr-\totalheight+\ht\strutbox\relax}%
        [\dimexpr\ht\strutbox+17mm][\dp\strutbox]{\expandafter\includegraphics[width=0.01cm]{#1}}%
}}

\usepackage[colorinlistoftodos,prependcaption,textsize=tiny,textwidth=30mm]{todonotes}

\begin{document}

\title{Fully Dynamic Approximate Minimum Cut in Subpolynomial Time per Operation}
\author{Antoine El-Hayek\thanks{Institute of Science and Technology Austria, Klosterneuburg, Austria. email: \tt{antoine.el-hayek@ist.ac.at}}
\and Monika Henzinger\thanks{Institute of Science and Technology Austria, Klosterneuburg, Austria. email: \tt{monika.henzinger@ist.ac.at}}
\and Jason Li\thanks{Carnegie Mellon University. email: \tt{jmli@cs.cmu.edu}}
}

\date{}

\maketitle

\fancyfoot[R]{\scriptsize{Copyright \textcopyright\ 2025 by SIAM\\
Unauthorized reproduction of this article is prohibited}}

\sloppy
\begin{abstract}
    Dynamically maintaining the minimum cut in a graph $G$ under edge insertions and deletions is a fundamental problem in dynamic graph algorithms for which no conditional lower bound on the time per operation exists. In an $n$-node  graph the 
    
    best known $(1+o(1))$-approximate algorithm takes $\tilde O(\sqrt{n})$ update time~\cite{thorup2007fully}. 
    If the minimum cut is  guaranteed to be $(\log n)^{o(1)}$, a deterministic exact algorithm with $n^{o(1)}$ update time exists~\cite{DBLP:conf/soda/JinST24}.
    
    We present the first fully dynamic algorithm for $(1+o(1))$-approximate minimum cut with $n^{o(1)}$ update time. Our main technical contribution is to show that it suffices to consider small-volume cuts in suitably contracted graphs.

\end{abstract}

\section{Introduction}

Computing the minimum cut in a graph is a fundamental problem in graph algorithms with near-linear randomized and deterministic algorithms in 
unweighted~\cite{kawarabayashi2015deterministic,DBLP:journals/siamcomp/HenzingerRW20} and weighted graphs~\cite{karger2000minimum,DBLP:conf/soda/HenzingerLRW24}. Its dynamic variant, where the minimum cut has to be maintained under edge insertions and deletions, is, together with the minimum spanning tree problem, one of the few dynamic graph problems for which no polynomial  lower bound, conditional or unconditional, on the time per update or query exists. For the minimum spanning tree problem no such lower bound can exist as the problem can be solved in polylogarithmic time per operation (see e.g.~\cite{DBLP:journals/jacm/HolmLT01}). However, for the dynamic minimum cut problem the question of its time complexity is still wide open:
No polylogarithmic-time dynamic minimum cut algorithm is known, except for the insertions-only setting~\cite{DBLP:journals/jal/Henzinger97,DBLP:journals/talg/GoranciHT18}. 

Thus it is not surprising that  a sequence of recent works has studied the fully dynamic version of the problem:
In an $n$-node graph Goranci, Henzinger, Nanongkai, Saranurak, Thorup, and Wulff-Nilsen~\cite{DBLP:conf/soda/GoranciHNSTW23} present the best known exact algorithms in general graphs. The first one is a Monte-Carlo randomized algorithm and takes $\tilde O(n)$ worst-case update time, the second one is deterministic and takes $\tilde O(m^{30/31})$ amortized update time.  If the minimum cut is guaranteed to be $(\log n)^{o(1)}$, Jin, Sun, and Thorup gave a deterministic exact algorithm with $n^{o(1)}$ worst-case time~\cite{DBLP:conf/soda/JinST24}.
In the approximation regime the classic $(1+o(1))$-approximate algorithm by Thorup is still the fastest known algorithm. 
It is a Monte-Carlo randomized algorithm and takes time $\tilde O(\sqrt{n})$ worst-case time per update operation~\cite{thorup2007fully}. 
Furthermore, Thorup and Karger~\cite{thorup2000dynamic} gave a deterministic $\sqrt{2 + o(1)}$-approximation algorithm with $O(\lambda^2 \polylog n)$ \ update time, where $\lambda$ is the minimum cut size.
All these algorithms take constant time to return the (exact, respectively, approximate) minimum cut size, and the algorithms in~\cite{thorup2007fully} and \cite{DBLP:conf/soda/JinST24} are able to output the minimum cut in time $\tilde O(\lambda)$, respectively, in time $n^{o(1)}$. Note that all these algorithms work on unweighted graphs.

In this paper we present another step towards resolving the time complexity of the dynamic minimum cut problem: We give the first fully dynamic algorithm for $(1+o(1))$-approximate minimum cut in an unweighted graph with $n^{o(1)}$ amortized update time that works for all minimum cut values. It is a Monte-Carlo randomized algorithm that returns the correct output with high probability. It can either output the cut value in constant time or the set $S$ of vertices on the smaller side of the cut in time $\tilde O(|S|)$. 

\textbf{Technical contribution.} We next describe the main technical challenges and how we overcome them. 

\emph{Master algorithm:} Our first step is to sparsify the graph as to ensure that the minimum cut  value is reduced to at most $O(\frac {\log n} {\epsilon ^2})$ using the edge-sampling approach of Karger~\cite{kargersparse}.
It samples each edge with a probability $p$ aptly chosen as described in the following lemma:

\begin{restatable}[Theorem 2.1 and Corollary 2.1 of~\cite{kargersparse}]{theorem}{kargersparse}\label{thm:kargersparsify}
    Let $G$ be any graph with minimum cut $\lambda$ and let $p \ge \frac {54 \ln n} {\epsilon^2 \lambda}$.
    Let $G(p)$ be the graph obtained from $G$ by sampling each edge independently with probability $p$.
    Then with high probability, the minimum cut $S$ in $G(p)$ will correspond to a $(1+\epsilon)$-minimum cut of $G$.
    Moreover, also with high probability, the value of \emph{any} cut $S$ in $G(p)$ has value no less than $(1-\epsilon)$ and no more than $(1+\epsilon)$ times its expected value. 
\end{restatable}

Ensuring that the minimum cut is small is crucial for our analysis as our running time depends on the size of the minimum cut. 
However, one challenge with sparsification is that it is hard to dynamize whenever the value of the minimum cut changes, and, thus, $p$ has to change as well. 
In that case, one update to the graph can lead to many updates to the sparsified graph, which we cannot afford.
To deal with this problem we run $O(\log n)$ many algorithms, each with a different value of $p$, in parallel, and we run a \emph{master algorithm} to decide which answer to return based on the answers of these algorithms.
In the following we use $\lambda$ to denote the minimum cut in the full graph\footnote{We will denote $\lambda$ the size of the minimum cut in the original graph, and $\tilde \lambda$ the size of the minimum cut in the sparsified graph. We further specify $\lambda_i$ to denote the size of the minimum cut in the sparsifed graph $G_i$ when necessary. $b_i$ is a fixed parameter of algorithm $i$.}.
Define $b_i = 1.1^i$ for all integers $i \in[1, \log_{1.1} n^2]$ and run one algorithm for each index $i$.
The $i$-th algorithm, i.e., the algorithm for index $i$, assumes that the value $\lambda$ of the minimum cut belongs to the range  $[b_i, b_{i+1})$ (called the \emph{range} of the $i$th-algorithm), and sparsifies with $p_i=\frac {54 \ln n} {\epsilon^2 b_i}$ to create graph $G_i$.
It follows that, if $\lambda \in [b_i, b_{i+1})$, then with high probability (w.h.p.) the minimum cut $ \lambda_i$ in the sampled graph $G_i$ is a $(1+\eps)$-approximate minimum cut in the full graph and w.h.p. its size in $G_i$ belongs to the range $[(54 (1-\eps) \ln n )/\eps^2, (54 \cdot 1.1 (1+\eps) \ln n)/\eps^2 ]$.
Thus, for each $i \in[1, \log_{1.1} n^2]$  the $i$-th algorithm calls a subroutine -- which is the central algorithm of our paper, we call it the \emph{bounded mincut algorithm} -- with two parameters $\lmin$ and $\lmax$, on $G_i$. The bounded mincut algorithm's requirements are:
\begin{enumerate}[noitemsep]
    \item If $ \lambda_i < \lmin:$ any output is acceptable.
    \item If $\lmin \le  \lambda_i \le \lmax$, output an approximation of the minimum cut.
    \item If $ \lambda_i > \lmax$, output ``$ \lambda_i > \lmax$''.
\end{enumerate}
As per our above discussion, we choose
$\lmin=(54 (1-\eps) \ln n )/\eps^2$ and $\lmax = (54 \cdot 1.1 (1+\eps) \ln n)/\eps^2$. 

Hence, out of all the $i$-th algorithms, at least one algorithm  outputs a correct approximation. There is, however,  only one $i$ such that  $\lambda \in [b_i, b_{i+1})$.
It thus remains to decide which answer out of the $O(\log n)$ many answers returned by the $O(\log n)$ algorithms should be returned by the master algorithm.

For that, we design the algorithms so that for every $i$ if $b_i \le \lambda$ (which corresponds w.h.p. to $\lmin \le  \lambda_i $), the $i$-th algorithm outputs a correct answer, more specifically it  either outputs a $(1+\epsilon)$-approximation of $\lambda$ or it outputs that the minimum cut in $G_i$ is larger than $\lmax$ (called the \emph{correctness requirement} of the algorithm). However, if $b_i > \lambda$ (which corresponds w.h.p. to $\lmin >  \lambda_i$), the algorithm can output an arbitrary value. The reason for this is as follows:
A $b_i$ value  that is much larger than $\lambda$ can lead to undersampling of the graph, and, thus, the minimum cut in $G_i$ can deviate too much from its expected value.
On the other side, a $b_i$ value that is smaller than $\lambda$ results -- in the worst case -- in oversampling, which is not a problem for correctness. 
The only issue that might arise is that the minimum cut $ \lambda_i$ in the sampled graph $G_i$ is larger than $\lmax$, but our algorithm is required to detect this. 
Since, by Theorem~\ref{thm:kargersparsify} and our choice of $\lmin$, $\lmax$, and $p_i$, it holds w.h.p.  for the unique $i$ with $b_i \le \lambda < b_{i+1}$ that (a) $\lmin\le \tilde \lambda \le \lmax $ and that (b) we did not undersample, the master algorithm would output an $(1+\epsilon)$-approximation of $\lambda$ if it output the value returned by the $i$-th algorithm. However, many algorithms might return a value (instead of returning just the fact $\lambda > \lmax$ (case 3 above)) and, thus, it is not possible for the master algorithm to determine $i$. This problem is resolved as follows:
As all algorithms $j$ with $b_j \le \lambda$ give a correct answer and, for $\epsilon \le 0.1$,  w.h.p. for only two such algorithms it can be the case that $ \lambda_j \le \lmax$ (by the choice of $\lmax$), namely for $j = i$ and $j = i-1$. Thus, at most two algorithms with $b_j \le \lambda$ return a value that is at most $\lmax$, all other such algorithms will either return a value that is larger than $\lmax$ or will state that $\lambda_j > \lmax$. 
Thus
the master algorithm takes the output of the algorithm with the smallest $j$ that did not output that $ \lambda_j > \lmax$ or return a value larger than $\lmax$. This guarantees that either the output of algorithm $i$ or algorithm $i-1$ is returned.

This then results in a $(1+\epsilon)^2$-approximation of $\lambda$.
The overlap in the ranges of the algorithms is needed to ensure that we catch all the corner cases.

\emph{The bounded mincut algorithm:} Below we describe the bounded mincut algorithm, i.e. any reference  to an algorithm is to the bounded mincut algorithm, not to the master algorithm.

We first describe how the  algorithm computes its output. For simplicity we  assume in this overview that, after sparsification,  the value of the minimum cut $ \lambda_i$ of the sparsified graph after every update is in $[\lmin, \lmax]$.
We first give the high level approach and afterwards explain it in more detail: The algorithm is given $\lmin$ and $\lmax$ and it relies on three key techniques: the dynamic expander decomposition, as presented by Goranci, Räcke, Saranurak and Tan~\cite{expanderhierarchy}, the concept of $(1-\epsilon)$-boundary sparseness, as presented by Henzinger, Li, Rao, and Wang~\cite{DBLP:conf/soda/HenzingerLRW24},
and a ``local'' version of Karger's minimum cut algorithm, which we call \emph{LocalKCut}, introduced by Nalam and Saranurak~\cite{kragerlocal}. We decompose the graph using the dynamic expander decomposition and call each resulting expander a \emph{cluster}. 
We then further partition each cluster along $(1-\eps)$-boundary sparse cuts (defined below), which we find using LocalKCut, until every cluster contains no such cut anymore. 
Then we collapse each cluster into a single node and we have the guarantee that the size of the minimum cut in the collapsed graph is at most a factor $(1+2\eps)$ larger than $ \lambda_i$. 
We run the algorithm recursively on the collapsed graph until only two nodes are left. 
This results in a hierarchy of collapsed graphs. The minimum of the cut in the final collapsed graph and of various cuts that we compute at different levels of the hierarchy (which we have not yet described) gives the desired output  of the $i$-th algorithm. This hierarchy, as well as the corresponding cuts, are maintained dynamically.

We next describe the vertex partition into clusters, as well as how it is created and maintained dynamically by the algorithm in more detail.
Let us first recall the definition of $(1-\epsilon)$-boundary-sparseness:

\begin{definition}[Boundary-Sparse. Def. 2.3 of~\cite{DBLP:conf/soda/HenzingerLRW24}]\label{def:sparse}
    For a set $C \subsetneq V$ and parameter $\epsilon \leq 1$, a cut $U \subseteq C$ is $(1-\epsilon)$-boundary sparse in $C$ iff 
    $$
    w(U, C\setminus U)< (1-\epsilon) \min\{w(U, V\setminus C), w(C\setminus U, V\setminus C)\}
    $$
\end{definition}
We say a cut $S \subsetneq V$ \emph{crosses} a vertex set $C$ if
both $S\inter C$ and $C \setminus S$ are nonempty.
The main idea behind boundary-sparseness is as follows: If a set $C$ contains no $(1-\epsilon)$-boundary sparse cuts, then any cut $S$ of the graph that crosses $C$ with $S\inter C=U$ can be approximated by a cut $S'$ that does not cross $C$, namely by the cut that is identically to $S$ outside of $C$ and that ``follows'' the boundary of $C$ with weight 
$\min\{w(U, V\setminus C), w(C\setminus U, V\setminus C)\}$ to avoid cutting $C$. 
Note that $w(S',V \setminus S') \le w(S,V\setminus S) / (1-\eps) \le w(S, V\setminus S)(1+2\eps),$  i.e., the size of $S'$ is at most a factor of $(1+2\epsilon)$ larger than the size of $S$. Thus, if $S$ is a minimum cut, then $S'$ is a $(1+2\eps)$ approximation of $S$.
We say that \emph{$S$ is uncrossed by $S'$ in cluster $C$.}
As $(1-\eps)$-boundary sparse cuts cannot be uncrossed, the algorithm partitions each cluster into two new clusters along $(1-\eps)$-boundary sparse cuts until no such cuts exist in any cluster, and then it contracts each remaining cluster  and  recurses on the contracted graph.
Each level of the recursion increases the approximation ratio by at most a factor of $(1+2\epsilon)$.
We ensure that the recursion depth is at most $r=\Theta(\log^{\frac 1 4} n)$, so that the total approximation ratio is $(1+2\epsilon)^r$, which is $1+o(1)$ when $\epsilon\le o(1/r)$.

\emph{Static cluster decomposition algorithm.}
As the algorithm does not need to detect the minimum cut if it has value larger than $\lmax$, the actual cluster creation works as follows:
Clusters are created by  first computing an expander decomposition and then  recursively cutting these initial clusters further along $(1-\epsilon)$-boundary sparse cuts \emph{of size at most $\lmax$}, until no such cuts exist anymore. 
To find these cuts we crucially rely on the expander property, namely that \emph{every cut that has cut-size at most $\lmax$ in the full graph, has cut-size at most $\lmax$ and volume at most $\frac \lmax \phi$ inside the $\phi$-expander.}
Very importantly this property continues to hold in any cluster that is a subset of a $\phi$-expander, even though the cluster might not be an expander. Thus using $\phi = n^{-o(1)}$ and $\lmax = O(\log^2 n)$ any cut of size at most $\lmax$ in such a cluster has small (enough) volume $n^{o(1)}$.

We can find \emph{all}  cuts of size at most $\lmax$ and of small volume in a cluster $C$ that is a subset of a $\phi$-expander efficiently by using the LocalKCut Algorithm, whose precise guarantee is deferred to the next section. We run LocalKCut from each vertex $v$ in $C$ to find all cuts of size at most $\lmax$ and volume at most $\frac \lmax \phi$ in $C$.
It follows from the above observation that all cuts of size at most $\lmax$ in $C$ will be found.

\emph{How to dynamize that algorithm:}
We next describe the main challenges we face to make this general approach  dynamic: 

(1)
Finding all boundary sparse cuts of size at most $\lmax$ while edges are inserted and deleted.

(2) As our hierarchy structure is recursive, we have to make sure that we do not get a multiplicative blowup in the number of update operations per hierarchy level, as well as the depth of our recursion. More specifically, one update on the top level of the hierarchical decomposition can lead to up to $h>1$ changes to the next level and, thus, $h^r$ many on the bottom level.
We show that for our decomposition, the depth $r$  of the recursion is $\log^{1/4}n$, and the number of changes is $h = 2^{O(\log^{1/2} n)}$, so that 
$h^r = 2^{O(\log^{3/4} n)} = n^{o(1)}$.

(3) We need to watch out for the case where the minimum cut uncrosses in a way that one side of the cut becomes the empty graph. 

(4) We need to maintain a data structure that enables us to reach and update the values and objects (nodes, edges, clusters) whenever needed.

(5) The algorithm described above works efficiently and correctly as long as $ \lambda_i$ belongs to the range of the algorithm. However, we also need to maintain the data structure when $ \lambda_i$ does not belong to this range as it might do so again in the future. 

We now explain how we address each of the above challenges:

(1) Boundary-sparse cuts were 
used in the static setting to compute a minimum cut in~\cite{DBLP:conf/soda/HenzingerLRW24}. The algorithm in~\cite{DBLP:conf/soda/HenzingerLRW24}  does not use an expander decomposition as it runs in deterministic time $\tilde O(m)$ and no deterministic algorithm for computing an expander decomposition is known in that running time. Thus, it  had to find   boundary-sparse cuts with arbitrarily large volume on either side of the cut, while we use the crucial observation that in a subset of an expander it is sufficient to
find \emph{unbalanced} boundary-sparse cuts, i.e.~cuts where the volume on one of the sides is small, as discussed above.
This allows us to deal with edge insertions and deletions as follows: 
We maintain the invariant that no cluster contains a $(1-\epsilon)$-boundary sparse cut of size at most $\lmax$.
Then, after each insertion or deletion, we only have to check \emph{locally } (i.e. at the endpoints of the updated edge) using the LocalKCut algorithm if the invariant still holds, and, if not,  enforce it by cutting any $(1-\epsilon)$-boundary sparse cut that LocalKCut may find.
The expander property here again is key, as the volume we need to explore is $\frac{ \lmax} {\phi } = \lmax \cdot n^{o(1)}$. Note that this is the first use of LocalKCut in the dynamic setting.

(2) By cutting in this way, however, we need to make sure that we do not cut too many edges, as this might either make the depth of our hierarchy too large, or the number of updates per level too high.
To deal with this, we use the fact that the dynamic expander decomposition of~\cite{expanderhierarchy} is recomputed from scratch every $\Theta(\frac m {n^{o(1)}})$ updates, called a \emph{restart} of the cluster decomposition. The number of updates is carefully chosen so as  to, on the one side, not end up with too many inter-cluster edges and, on the other side, amortize the running time accrued since the last restart over these updates.
Using the restarts and a potential function argument, we can bound  the number of inter-cluster edges per level at any point in time to be only a factor $f(\eps)$ (where $f$ is a monotonically increasing function) larger than the sum of the number of \emph{inter-expander} edges plus the recourse of the expander decomposition  at the same point in time. This implies that the number of inter-cluster edges decreases  suitably when going from one level to the next, and, thus, limits the depth $r$ of our hierarchy to $\log^{1/4}n$. We can also bound the amortized number of updates per level by $2^{O({\log^{1/2} n})}$.

(3) To make sure that we catch the minimum cuts that might uncross to the empty cut, we build, for each cluster $C$, a copy of the graph where all the nodes outside of the cluster are contracted into one node.
We call it the \emph{mirror cluster} of cluster $C$.
We can show that if a minimum cut uncrosses to the empty cut for a cluster $C$, then the mirror cluster of $C$ must have a cut of similar size (a $(1+2\epsilon)$-factor away at most) and volume $n^{o(1)}$ 
and, thus, can be found efficiently using LocalKCut.
Note that we also have to maintain the mirror clusters dynamically. These graphs were not used in prior minimum cut algorithms.

(4) Our data structure essentially keeps the cluster decomposition, and along with each cluster, its corresponding mirror cluster, its set of outgoing edges and the number of outgoing edges.
To maintain this data structure efficiently after edge updates, we guarantee that between two consecutive restarts,  we only refine our decomposition, i.e., we \emph{do not merge clusters, even if the expander decomposition would do so}.
Indeed, the expander decomposition might need to merge clusters back together to maintain the expansion property, but since we only need to maintain that each cluster is a subset of an expander, merging back is unnecessary.
Since the expander decomposition creates $\Tilde O(m\phi)$ inter-cluster edges on each rebuild, and  has amortized recourse\footnote{The recourse of an algorithm is the number of changes to the solution. In the case of the expander decomposition, the algorithm maintains which edges are intercluster edges and which ones are intracluster edges. The recourse is thus the number of edges per update that switch from being intercluster ones to intracluster ones and vice-versa.} $\rho=2^{\Theta(\sqrt {\log n})}$ per edge update over a period of $O(\frac{m\phi}{\rho})$ updates until the next rebuild, the number of inter-expander edges that exist at any point between two consecutive rebuilds is $\Tilde O(m\phi)$.
Hence the number of intercluster edges in the refinement of the expander decomposition that our algorithm maintains is $\Tilde O(m\phi)$.
Our data structure allows us to split a cluster $C$ in time linear in $\vol(S)$, where $S$ is the side of the cut of smaller volume. 
As each edge can be on the smaller  side of the cut only $O(\log m)$ times, and the total time for maintaining our structure is $\Tilde O (m)$,  we can amortize that time over $\Theta(\frac{m\phi}{\rho})$ updates and achieve an $n^{o(1)}$ amortized time per update.

(5) Let us now discuss what happens if the minimum cut $ \lambda_i$ becomes larger than $\lmax$ or smaller than $\lmin$.
In the first case, we need to output the value of a cut that is larger than $\lmax$ or state that no cut of size at most $\lmax$ exists, and also maintain the data structure so that if $ \lambda_i$ becomes smaller than $\lmax$ in the future, our algorithm is again able to fulfill its requirement. 
To solve those issues, we introduce a boolean variable at each vertex $v$ to remember whether or not LocalKCut should be run from  $v$, that states whether a node is \emph{checked} or \emph{unchecked}, and enforce the invariant below:
\begin{restatable}{invariant}{cuttingisgood}\label{inv:cuttingisgood}
For every cluster $C$
    and every cut $S \subsetneq C$ such that $S$ is $(1-\epsilon)$-boundary-sparse in $C$, $\vol (S) \le \frac \lmax \phi$, and $\lmin \le w(S, C\setminus S) \le \lmax$, there exists a node $v \in S$ such that $v$ is unchecked.
\end{restatable}
Intuitively, a vertex is unchecked if one should run LocalKCut on it to find a potential $(1-\epsilon)$-boundary-sparse cut it might be included in.
All vertices start as unchecked, and become checked when we run LocalKCut on them and it doesn't return a suitable cut.
A checked node becomes unchecked if it is incident to an edge update.
The idea is to output a solution only if all nodes are checked, but as we will discuss below, it might not always be the case.
Running our algorithm as described above, we cannot ensure that the current output of the algorithm is the  value of a minimum cut, as there might exist a cluster $C$ with a $(1-\epsilon)$-boundary sparse cut, whose value is above $\lmax$ and our algorithm might not find such cuts.
However, since it outputs the value of a cut, the returned value will be at least $\lmax$, satisfying the correctness requirement.
Moreover, since \Cref{inv:cuttingisgood} still holds throughout, it ensures that when $ \lambda_i$ drops below $\lmax$, our algorithm outputs a $(1+2\epsilon)$-approximation once again.

The situation is more difficult in the case where $ \lambda_i$ becomes smaller than $\lmin$, as our potential function argument relies on the fact that 
every cluster that we create has 
boundary size at least $\lmin$. 
To make sure that this is true throughout the algorithm, any time that we find a $(1-\eps)$-boundary sparse cut of size at most $\lmax$ with boundary size\footnote{Here, the cut size refers to the size of the cut in the considered cluster, while the boundary size is the size of the cut in the whole graph.} strictly less than $\lmin$ in a cluster $C$, we do not split $C$, but instead \emph{freeze} $C$.
Freezing a cluster means that we do not try to run LocalKCut on unchecked vertices of the cluster, but instead we have to wait until the cluster becomes unfrozen again. This saves running time and recourse, i.e. changes in the cluster decomposition, which would lead to changes in the next level of hierarchy. 
This does not affect correctness, as it is guaranteed that $ \lambda_i < \lmin$ (i.e. $b_i > \lambda$) in that case, i.e., the master algorithm will ignore the output of the $i$-th algorithm.
We then maintain the boundary size of that cut under all further edge updates until it reaches $\lmin$ again. Then we unfreeze $C$ and execute all updates in $C$ that happened since it was frozen. To be able to do this we store all  updates that have at least one endpoint in a frozen cluster at each frozen cluster.
That way, when $ \lambda_i$ reaches $\lmin$, we ensure that all clusters of the $i$-th algorithm are unfrozen and thus all updates are processed, and the potential function argument goes through again.

Combining this leads us to the following theorem:

\begin{restatable}{theorem}{maintheorem}
    There exists a randomized dynamic algorithm that maintains a $(1+o(1))$-approximate minimum cut over unweighted graphs in $n^{o(1)}$ amortized time.
    The set of nodes $S$ representing the cut is stored implicitly and can be output in $\Tilde O(\card S)$ time.
\end{restatable}

\subsection{Organization}
The paper is organized as follows: in \Cref{sec:prelims}, we define the problem and introduce definitions and notations. 
In \Cref{sec:cdecomposition}, we introduce the notions of cluster decomposition and cluster hierarchy, central to our analysis.
 In \Cref{sec:static}, we present a static version of the bounded mincut algorithm, that can solve the problem for a mincut of value $ \lambda_i \in [\lmin, \lmax]$ which will be the basis for the dynamic algorithm, presented in \Cref{sec:dynamic}.
Then, in \Cref{sec:combining}, we show how we use the dynamic bounded mincut algorithm in the master algorithm to solve the problem for any value of $\lambda$.
Finally, in \Cref{sec:karger}, we present some new results on LocalKCut.

\section{Preliminaries}\label{sec:prelims}
\paragraph{Problem Definition}
We are given an undirected unweighted graph $G=(V,E)$, where $V$ is the set of vertices, and $E$ is the set of edges. Every $\emptyset \subsetneq S \subsetneq V$ is a \emph{cut} of $G$. We denote by $\vol(S)$ the sum of the degrees of the nodes in $S$.
We define, for every $T \subseteq V$, a cost associated with $T$, that we call the \emph{boundary} of $T$, or equivalently, the \emph{(cut) size} of $T$:
$$\boundary T = \card{E(T, V\setminus T)}$$

A \emph{minimum cut} of $G$ is a nonempty set $MC(G) \subsetneq V$ that minimizes the cut: $MC(G) \in \argmin_{\varnothing \subsetneq T \subsetneq V}\PAR{\boundary T}$.

We consider the \emph{$(1+\epsilon)$-approximate minimum cut} problem: given a graph $G$ as described above, find a set $S$ such that $\boundary S \le (1+\epsilon) \boundary MC(G)$.

We study the problem in the \emph{dynamic setting under edge updates}, that is, at every time step, an edge can be either inserted or deleted.

\begin{remark}
    We allow parallel edges. 
All our running times that include the number of edges count their multiplicities, that is, for example, in a graph with two nodes and 3 parallel edges between them, we have that $m=3$.
\end{remark}

\paragraph{Definitions and Notations}

For any sets $A,B \subseteq V$, with $A\inter B = \varnothing$, we use $E(A,B)$ to denote the set of edges going from $A$ to $B$ and $w(A,B)$ to denote  $\card{E(A,B)}$.

For any subset $C \subseteq V$, $G/C$ is the graph where the set $C$ is contracted to a single node, adding parallel edges if necessary. No self-loops are allowed.
\begin{definition}[Crossing sets]
       Let $H=(V_H, E_H)$ with $V_H \subseteq V$ and $E_H \subseteq E$, be a subgraph of $G$.
     We define  the boundary   $\boundary_H T$ of a set $T \subseteq V_H$ to be the boundary of $T$ in $H$.

    We say that a cut $S$ \emph{crosses} cluster $C$ if both $S\inter C$ and $C \setminus S$ are nonempty.
\end{definition}

We use the concept of uncrossing as follows: we say that we \emph{uncross a cut $S$ in a cluster $C$} if $S$ crosses $C$ and we replace $S$ with  either $S\union C$ or $S \setminus C$.
We will only uncross when $S\inter C$ is not $(1-\epsilon)$-boundary sparse (see Definition~\ref{def:sparse}) in $C$, so that $S'$ is a $(1+O(\epsilon))$-approximation of $S$. To find suitable clusters we use the following notion of extreme sets.

\begin{definition}[Extreme set]
    A set $\varnothing \subsetneq S \subsetneq V$ is called an extreme set if every strict subset $T \subsetneq S$ satisfies $\boundary T > \boundary S$. 
\end{definition}

\begin{definition}
Let $0<\gamma < 1$.
    A set $\varnothing \subsetneq S \subsetneq V$ of a graph $G=(V,E)$ is said to be  $\gamma$-extreme, if every strict subset $\varnothing \subsetneq T \subsetneq S$ satisfies $\boundary T  > \gamma \cdot \boundary S$.
\end{definition}

In other words, $S$ is $\gamma$-extreme if no strict subset of it has a $\gamma$-factor larger boundary than $S$.

\begin{definition}[Def 1.1 of~\cite{kragerlocal}]
    For any vertex $x \in V$, we say $X$ is a $(x, \nu, \sigma, k)$-set if $x \in X$, $\vol(X) < \nu$, $\card{X} < \sigma$, and $\boundary X<k$. If $X$ is also an extreme set, then $X$ is called an $(x, \nu, \sigma, k)$-extreme set.
\end{definition}

We show in \Cref{sec:karger} the following result, as a slight extension of the LocalKCut algorithm of Nalam and Saranurak~\cite{kragerlocal}:
\begin{restatable}{lemma}{LocalKCut}\label{lem:localkcut}
Let $G$ be a simple graph.
The LocalKCut algorithm satisfies the following properties:
    \begin{enumerate}
        \item The total running time of LocalKCut with parameters $G,v,\nu, k$ is $\Tilde O (\nu)$.
        \item   Let $S$ be an extreme set of $G$, of cut value at most $k$ and volume $\vol(S) \le \nu$.
    Let $v$ be a vertex of $S$.
    Then LocalKCut$(G,v,\nu, k)$ outputs $S$ with probability  $\Omega(\frac 1 {\card S^2})$.
    \item Let $S$ be a connected $\gamma$-extreme set of cut value $c \le k$ and volume $\vol(S) \le \nu$ in $G$.
    Assume $\gamma$ is the inverse of an integer.
    Let $v$ be a vertex of $S$.
    Then LocalKCut$(G,v,\nu, k)$ outputs $S$ with probability $\Omega(\card S ^{ -\frac 2 \gamma} c^{-\frac 2 \gamma+2}) $.
    \end{enumerate}
\end{restatable}
We call $v$ the \emph{starting vertex} of LocalKCut if $v$ is the second parameter in the call to LocalKCut.
\begin{definition}
For any subset $U \subseteq V$, $G[U]$ is the subgraph induced by $U$, and $G[U]^r$ for $r \in \R$ is the subgraph induced by $U$ where for every $v \in U$, we add $\ceil{r}$ self-loops for every boundary edge $\{v, x\}, x \notin U$ that is incident to $v$ in $G$.
\end{definition}
\begin{definition}
    The conductance of a cut $U \subsetneq V$ in the graph $G$ is $$\phi_G(U) = \frac{w(U, V\setminus U)}{\min\{\vol_G(U), \vol_G(V\setminus U)\}}.$$

A graph $G$ is a $\phi$-expander if every cut in $G$ has conductance at least $\phi$.
\end{definition}
\begin{definition}[$(\alpha, \phi)$-Expander, Definition 4.1 of~\cite{expanderhierarchy}]
    For a graph $G = (V, E)$ and parameters $\alpha, \phi \in (0,1)$,
a subgraph $U\subset V$ is $(\alpha, \phi)$-boundary-linked expander in $ G$ if the graph $ G[U ]^{\frac \alpha \phi}$ is a $\phi$-expander.
\end{definition}

An \emph{inter-expander} edge is an edge that lies between two different expanders. As we will have a cluster decomposition that is a refinement of the expander decomposition, any edge  that lies between two different clusters is an \emph{inter-cluster} edge.

 Finally, we will use the following result from Goranci, Räcke, Saranurak and Tan~\cite{expanderhierarchy}:

\begin{theorem}[Dynamic Expander Decomposition, Lemma 7.3 of~\cite{expanderhierarchy}]\label{thm:expanderdecomposition}
Suppose a graph $G$ initially contains $m$ edges and undergoes a
sequence of at most $O(\frac {m\phi}{\rho})$ adaptive updates such that $V (G) \le n$ and $E(G) \le m$ always hold.
Then there exists an algorithm that maintains an $(\alpha, \phi)$-expander decomposition with slack $38h$ and its contracted graph with the following properties:
\begin{enumerate}[noitemsep]
 \item update time: $\Tilde O (\frac {\psi 38^{2h}}{\phi ^2})$
\item preprocessing time: $\Tilde O(\frac m \phi)$
\item initial volume of the contracted graph (after preprocessing): $\Tilde O(\phi m)$
\item amortized recourse (number of updates to the contracted graph): $O(\rho)$
\end{enumerate}
\end{theorem}

In \Cref{thm:expanderdecomposition}, the contracted graph refers to the graph obtained from the original graph by merging all nodes in an expander into one (for each expander).
The amortized recourse is the number of changes made to the contracted graph, which correspond to the number of edges that switch from being 
intracluster edges 
to intercluster edges and vice-versa.
Since the amortized recourse is $O(\rho)$, and the algorithm can handle $O(\frac {m\phi} \rho)$ updates, this shows that overall the number of edges that were intercluster edges during that sequence of updates are at most $O(m\phi)$.

Throughout the paper we assume $\epsilon = \frac 1 {\sqrt{\log n}} \le 0.04$. We set $\lmin=(54 (1-\eps) \ln n )/\eps^2$ and $\lmax = (54 \cdot 1.1 (1+\eps) \ln n)/\eps^2$, $\phi = 2^{-\Theta (\log^{3/4}n)}, \rho= 2^{\Theta(\log ^{1/2}n)}$ and $ \alpha=\frac 1 {\poly \log n}.$

\section{Cluster Decomposition and Cluster Hierarchy}
\label{sec:cdecomposition}
In this section we introduce the cluster decomposition, the cluster hierarchy, as well as mirror clusters.
They are maintained in the central data structure of our dynamic algorithm, in \Cref{sec:datastructure}. 
The cluster decompositions, hierarchies and mirror clusters depend on the specific choice of $\epsilon, \lmax$ and $\lmin$.
Since these values do not change throughout the paper, we omit repeating them, but one should keep in mind that $\lmax \le 1.2\lmin$ and $\epsilon < 0.04$.
Remember that this section and all sections up to \Cref{sec:combining} (excluded) describe the bounded mincut algorithm which runs on the sampled graph $G_i$, but we will drop the index $i$ through this section. Thus, e.g., the value of the minimum cut of $G$ is denoted $\Tilde \lambda$.

We remind the reader of the requirements this algorithm: It has as input a dynamic (sparsified) graph $G$ with edge updates, and two fixed values $\lmin$ and $\lmax$.
\begin{enumerate}[noitemsep]
    \item If $\tilde \lambda < \lmin$, any output is acceptable.
    \item If $\lmin \le \tilde \lambda \le \lmax$, output an approximation of the minimum cut.
    \item If $\tilde \lambda > \lmax$, output ``$\tilde \lambda > \lmax$''.
\end{enumerate}

\subsection{Cluster Decomposition}
\begin{definition}[Cluster Decomposition]\label{def:clusterdecomposition}
Let $\ell$ be a positive integer. An $(\alpha, \phi, \lmax, \lmin)$-\emph{cluster decomposition} of a graph $G$ is a partition $\{C_j\}_{j \in [\ell]}$ of the vertex set of $G$. Moreover, if $\lmin \le \Tilde \lambda \le \lmax$, that partition satisfies, for each $j \in [\ell]$:
    \begin{enumerate}[label=\roman*., noitemsep]
        \item $C_j$ is contained in an $(\alpha, \phi)$-boundary linked expander.
        \item $C_j$ contains no $(1-\epsilon)$-boundary sparse cut $S$ with $\boundary_G(S) \le \lmax$ and $\vol(S) \le \frac \lmax \phi$.
    \end{enumerate}
\end{definition}
As $\alpha$, $\phi$, $\lmax$ and $\lmin$ do not change, we simply use \emph{cluster decomposition} to denote an $(\alpha, \phi, \lmax,  \lmin))$-cluster decomposition in the following.

The main point of the cluster decomposition is that it can approximate most minimum cuts, that is, in most cases, if $S$ is a minimum cut in $G$ of cut-size $\Tilde \lambda$ with $\lmin \le \Tilde \lambda \le \lmax$, then we can uncross $S$ by a cut $S'$ that consists of the union of some clusters, and is an approximate mincut.

Let us now talk about the corner cases. 
In fact, given a cluster decomposition we can always uncross a cut $S$ into a set $S'$.
However, $S'$ is not guaranteed to be a cut, that is, we could have $S'=\varnothing$ or $S' = V$.
In this case, we show $S$ is a \emph{local cut}, defined next.

\begin{definition}
    A \emph{local cut} is a cut of volume at most $4\frac \lmax \phi$.
\end{definition}

Local cuts are central in our analysis, as they are easier to find and maintain than non-local cuts.
To deal with the local cuts, we introduce \emph{mirror clusters}, one for each cluster.

\begin{definition}
    The mirror cluster, $C'$ of a cluster $C$ in graph $G$ is the graph $G/(V\setminus C)$, that is, the graph where all nodes outside of $C$ have been contracted into one.
\end{definition}

Mirror clusters have the nice property that all local cuts can be found by running LocalKCut on them.
We can thus state the main property of a cluster decomposition:

\begin{restatable}[Cluster Decomposition Uncrossing]{proposition}{clusterdecomposition}\label{lem:clusterdecomposition}
    Let $\{C_j\}_{j \in [\ell]}$ be an $(\alpha, \phi, \lmax, \lmin)$-cluster decomposition of a graph $G$.
    If $\lmin \le {\Tilde \lambda}\le \lmax$, then every minimum cut $S$ of value ${\Tilde \lambda}$ can be $(1+2\eps)$-approximated by a cut $S'$, that can be one of two cases:
    \begin{enumerate}[noitemsep]
        \item $S'$ is a cut of $G$ and $S'$ crosses no cluster.
        \item There exists a cluster $C_j$ such that $S'$ is a  local cut in the mirror cluster of $C_j$.
    \end{enumerate}
\end{restatable}
We prove this proposition in the rest of section.

We let $\{C_j\}_{j \in [\ell]}$ be a cluster decomposition be as defined in \Cref{lem:clusterdecomposition}, and first show that any minimum cut can cross at most two clusters using the following lemmata.

\begin{lemma}\label{lem:smallvolume}
    Let $S$ be a minimum cut of $G$, and let $C_j$ be a cluster that $S$ crosses.
    Let $X$ equal $ S\inter C_j$ if $\vol(S\inter C_j) \le \frac {\vol (C_j)} 2 $, and let $X$ equal $C_j \setminus S$ otherwise.
    If ${\Tilde \lambda} \le \lmax$ then $\vol(X) \le \frac {\lmax} \phi$.
\end{lemma}

\begin{proof}
    Let $C$ be the expander such that $C_j \subseteq C$.
    All volumes in this proof are the volumes of the sets in $G[C]^{\frac \alpha \phi}$, which are an upper bound of the volumes in $G$.
    We have that $w(S\inter C, C\setminus S) \le {\lmax}$, since $S$ is a minimum cut.
    
    Let $X' = S\inter C$ if $\vol(S\inter C) \le \frac {\vol (C)} 2 $, and $X'= C \setminus S$ otherwise.
    Since $C$ is an expander, we have that $\vol(X') \le \frac {\lmax} \phi$.
    Therefore, $\vol(X' \inter C_j) \le \frac {\lmax} \phi$.
    Recall that $\vol(X) = \min\{\vol(S \inter C_j), \vol(C_j \setminus S)\}$.
    If $X' = S\inter C$ then
    $\vol(X) \le \vol(S \cap C_j) \le \vol(S\inter C) = \vol(X') \le \frac {\lmax} \phi. $
    If $X' = C \setminus S$ then
    $\vol(X) \le \vol(C_j \setminus S ) \le \vol(C \setminus S) = \vol(X') \le \frac {\lmax} \phi. $
    Thus, the result follows.
\end{proof}

\begin{lemma}\label{lem:atmosttwo}
    In an $(\alpha, \phi, \lmax)$-cluster decomposition, any minimum cut $S$ of cut-size at most $\lmax$ can cross at most $2$ clusters.
\end{lemma}

\begin{proof}
    Recall that $\epsilon \le \frac 1 3$.
    Assume by contradiction that $S$ crosses three clusters, and wlog assume it is the clusters $C_1$, $C_2$ and $C_3$.
    Since $S$ is a minimum cut, we know that $w(S, V\setminus S)\le {\lmax}$.
    As the clusters are vertex disjoint, $w(S\inter C_1, C_1 \setminus S)+w(S\inter C_2, C_2 \setminus S)+w(S\inter C_3, C_3 \setminus S) \le w(S, V\setminus S) $ and thus there exists a $C_i$ such that $w(S\inter C_i, C_i \setminus S)\le \frac {w(S, V\setminus S)} 3$.
    Since $C_i$ contains no $(1-\epsilon)$-boundary-sparse cuts of cut-size at most ${\lmax}$ and of volume at most $\frac \lmax \phi$, and as by \Cref{lem:smallvolume}, $S\inter C_i$ has volume at most $\frac \lmax \phi$, we have that $S\inter C_i$ is not $(1-\epsilon)$-boundary-sparse:
    $$
    \frac {w(S, V\setminus S)} 3 \ge w(S\inter C_i, C_i\setminus S) \ge (1-\epsilon) \min\{w(S\inter C_i, V\setminus C_i), w(V \setminus (S\inter C_i), V\setminus C_i)\}
    $$
    And hence either $w(S\inter C_i, V\setminus C_i) \le \frac {w(S, V\setminus S)} {3(1-\epsilon)} \le \frac {w(S, V\setminus S)} 2 $ and $S\inter C_i$ is a cut of size at most $\frac {w(S, V\setminus S)} 3 + \frac {w(S, V\setminus S)} 2$, a contradiction to $S$ being a minimum cut,
    or $w(C_i\setminus S, V\setminus C_i) \le \frac {w(S, V\setminus S)} {3(1-\epsilon)} \le \frac {w(S, V\setminus S)} 2 $ and $C_i\setminus S$ is a cut of size at most $\frac {w(S, V\setminus S)} 3 + \frac {w(S, V\setminus S)} 2$, also a contradiction to $S$ being a minimum cut.
\end{proof}

\begin{corollary}\label{cor:1}
    Let $S$ be a minimum cut of $G$ of cut-size at most $\lmax$ and consider an $(\alpha, \phi, \lmax)$-cluster decomposition.
    We have 3 cases:
    (1) $S$ crosses no cluster.
    (2) $S$ crosses exactly one cluster. \label{case1}
    (3) $S$ crosses exactly two clusters.  \label{case2}
\end{corollary}

If we are in Case~(1) of \Cref{cor:1}, then we trivially are in Case~1. of \Cref{lem:clusterdecomposition} by simply setting $S'=S$. If on the other hand we are in either case (2) or (3), we need to uncross $S$ to get $S'$. Next we show that in those two cases there exists  a cut $S'$ of cut size at most $(1+2\eps) {\Tilde \lambda} $ that crosses no cluster. The proofs exploit the definition of $(1-\eps)$-boundary sparseness and \Cref{lem:smallvolume}.

\begin{lemma}\label{lem:3.6}
    If we are in Case (2) of \Cref{cor:1}, then there exists a cut $S'$ that is a $(1+2 \epsilon)$-approximation of $S$, and satisfies either:
    \begin{enumerate}[label=\roman*., noitemsep]
        \item $\vol_{G/(V\setminus C_1)}(S') \le 2\frac {\lmax} \phi + 2{\lmax}$, where $C_1$ is the cluster that $S$ crosses, and $S'$ is a cut of $G/(V\setminus C_1)$, or \label{casei}
        \item $S' \subsetneq V$ crosses no cluster and $S'$ is a cut of $G$. \label{caseii}
    \end{enumerate}
\end{lemma}

\begin{proof}
    Since $S$ crosses exactly one cluster, say $C_1$, and since $C_1$ contains no $(1-\epsilon)$-boundary sparse cut, we have that $w(S\inter C_1, C_1 \setminus S) \ge (1-\epsilon)\min\{w(C_1 \inter S, V\setminus C_1), w(C_1 \setminus S, V\setminus C_1)\}$.
    Let $X$ be the set among $C_1 \inter S$, $C_1 \setminus S$ that satisfies 
    $w(X, V\setminus C_1) = \min\{w(C_1 \inter S, V\setminus C_1), w(C_1 \setminus S, V\setminus C_1)\}$.
    Note that $w(X, V\setminus C_1) \leq w(S\cap C_1, C_1\setminus S)/(1-\eps) = w(X,C_1 \setminus X)/(1-\eps)$, which implies that
    $w(X, V\setminus C_1) - w(X,C_1 \setminus X) \le \eps \cdot w(X, V\setminus C_1) \le (\eps/(1-\eps)) w(S\cap C_1, C_1\setminus S) \le (3\eps/2) \cdot w(S\cap C_1, C_1\setminus S) $.

    Let's now temporarily set $S' = S \union X$ if $X \not\subset S$, or $S' = S \setminus X$ otherwise, i.e., we ``uncross'' $S$.
    Then we have that $w(S', V\setminus S') \le w(S, V\setminus S) - w(X, C_1 \setminus X) + w(X, V\setminus C_1) \le 
    w(S, V\setminus S) +  3\epsilon/2 \cdot w(S\cap C_1, C_1\setminus S)  \le
    (1+3\eps/2) w(S, V\setminus S)$.
    
    If $S'$ is a cut (both nonempty and a strict subset of $V$), we are in Case~\ref{caseii}

    If $S'\in \{\varnothing, V\}$, then either $S=X$ (and $S'=\emptyset$) or $S=V \setminus X$ (and $S' = V$).
    We assume wlog the former.
    By \Cref{lem:smallvolume}, we know that either $X$ or $C_1 \setminus X$ has volume at most $\frac {\lmax} \phi$. 
    Let $Y$ be the set of smallest volume of the two, which implies that $\vol(Y) \le \frac {\lmax} \phi$. 

    If $Y=X$, then set $S' = X$, and Case~\ref{casei} follows. 
    
    Otherwise, set $S' = (C_1 \setminus X) \union \{v\}$, where $v$ is a node that is created by contracting all vertices in $V\setminus C_1$.
    Then $v$ has at most $\vol(C_1 \setminus X)$ edges to $(C_1 \setminus X)$ and $\frac 1 {1-\epsilon} \lmax \le 2{\lmax}$ many to $X$.
    Therefore $\vol(S') \le 2\frac {\lmax} \phi + 2{\lmax}$, and we are in Case~\ref{casei}
\end{proof}

\begin{lemma}\label{lem:3.7}
    If we are in Case 3 of~\Cref{cor:1}, then there exists a cut $S'$ of either $G/(V\setminus C_1)$ or $G/(V\setminus C_2)$ that is a $(1+\frac 3 4 \epsilon)$-approximation of $S$, and satisfies $\vol_{G/(V\setminus C_i)} (S') \le 2\frac {\lmax} \phi + 2{\lmax}$ (where $j \in\{1,2\}$ depending on where $S'$ exists).
\end{lemma}

\begin{proof}
Let ${\Tilde \lambda}=w(S, V\setminus S)$.
    We have that $w(C_1\inter S, C_1 \setminus S) + w(C_2\inter S, C_2 \setminus S) \le w(S, V\setminus S) = {{\Tilde \lambda}}$.
    Assume thus wlog that $w(C_1\inter S, C_1 \setminus S) \le \frac {{\Tilde \lambda}} 2$.

    Since $C_1$ contains no $(1-\epsilon)$-boundary sparse cuts, we have that $w(S\inter C_1, C_1 \setminus S) \ge (1-\epsilon)\min\{w(C_1 \inter S, V\setminus C_1), w(C_1 \setminus S, V\setminus C_1)\}$.
    Let $X$ be the set among $C_1 \inter S$, $C_1 \setminus S$ that satisfies 
     $w(X, V\setminus C_1) = \min\{w(C_1 \inter S, V\setminus C_1), w(C_1 \setminus S, V\setminus C_1)\}$.
    It follows that 
    $w(X, C_1 \setminus X) = w(S\inter C_1, C_1 \setminus S) \ge (1-\epsilon)w(X, V\setminus C_1)$.
    We thus have that $w(X, V\setminus X) = w(X, V\setminus C_1) + w(X, C_1\setminus X) \le ((1/(1-\epsilon) + 1) w(X, C_1\setminus X) \le ({{2-\epsilon}\over {2(1-\epsilon)}}) {{\Tilde \lambda}}  = (1+ {\epsilon \over 2-2\epsilon}) {\Tilde \lambda} \le (1+3\eps/4){\Tilde \lambda}$. 

    We know by \Cref{lem:smallvolume} that either $X$ or $C_1 \setminus X$ have volume at most $\frac {\lmax} \phi$. 
    Let $Y$ be the set of smallest volume of the two, which implies that $\vol(Y) \le \frac {\lmax} \phi$.

    If $Y=X$, then set $S' = X$, and the result follows.
    Otherwise, set $S' = (C_1 \setminus X) \union \{v\}$, where $v$ is the contracted $V\setminus C_1$.
    Then $v$ has at most $\vol(C_1 \setminus X)$ edges to $(C_1 \setminus X)$ and $2{\lmax}$ many to $X$.
    Therefore $\vol(S') \le 2\frac {\lmax} \phi + 2{\lmax}$, and we get the result.  
\end{proof}

Note that both lemmata show the existence of a \emph{cut} $S'$, i.e., $S'$ is required to be both nonempty and different from $V$.
This concludes the proof of \Cref{lem:clusterdecomposition}.

\subsection{Cluster Hierarchy}

The $(\alpha, \phi, \lmax)$-cluster decomposition and the mirror clusters are  the building blocks of the cluster hierarchy, defined as follows:

\begin{definition}[Cluster Hierarchy]\label{def:clusterhierarchy}
    Let $h$ be a positive integer. An $h$-cluster hierarchy is a set of graphs $\tilde G_1, \dots, \tilde G_h$ with $G_1=G$
    with $\vol({\tilde G}_h) \le \frac \lmax \phi$, together with $h$ many $(\alpha, \phi, \lmax, \lmin)$-cluster decompositions, one for each ${\tilde G}_j, j\in [h]$, and a mirror cluster for each cluster of each decomposition. 
    Moreover, each ${\tilde G}_j, j>1$ is obtained from ${\tilde G}_{j-1}$ by contracting each cluster.
\end{definition}

The main point of the $h$-cluster hierarchy is to ensure that we only have to maintain local cuts and their sizes, in other words, to transform any minimum cut into a local cut, if not on this level of the hierarchy then further down in the hierarchy.

We will use the $h$-cluster hierarchy as follows: Let $S$ be a minimum cut of cut-size in $[\lmin, \lmax]$. 
Assume for the moment, that we have access to all local cuts.
By \Cref{lem:clusterdecomposition}  there exists in ${\tilde G}_1$ an approximate cut $S'_1$ that is either local or that crosses no cluster.
In the former case, we are done. In the latter, $S_1'$ induces a cut in ${\tilde G}_2$. 
Again by \Cref{lem:clusterdecomposition} there exists in ${\tilde G}_2$ an approximation $S'_2$ of $S_1'$ that is either local or that crosses no cluster.
Again, if it is local, we are done. If it is not, then we move on to ${\tilde G}_3$. 
Following this logic, we arrive the following proposition:

\begin{proposition}\label{prop:clusterhierarchy}
     Let $h$ be a positive integer. In an $h$-cluster hierarchy of $G$, there exists a local cut $S'$ in a mirror cluster $C$ of the cluster decomposition of a collapsed graph $\tilde G_j$ or in the final collapsed graph, that is a $(1+2\epsilon)^h$ approximation of the minimum cut.
    Moreover, there exists a corresponding cut $S$ in $G$ such that $S$ and $S'$ have the same cut size.
\end{proposition}

\begin{proof}
    It suffices to apply \Cref{lem:clusterdecomposition} iteratively starting at $G$ until we end up in Case 2.\ of that proposition. 
    If we never arrive in Case 2, it follows that the recursion reaches the graph ${\tilde G}_h$, which has volume $\frac\lmax \phi$. Thus, all its cuts are local.
    Once a local cut is found, to find the corresponding cut in $G$, since all graphs (including mirror clusters) in the hierarchy are contractions of graphs from the level above, it suffices to uncontract the graph.
    Uncontracting all the way to the top gives the result.
    The approximation value comes from the fact that on each level, the size of the cut increases by a factor of at most $(1+2\epsilon)$.
\end{proof}
The goal of further sections will thus be to compute and maintain such a cluster hierarchy.

\subsection{Mirror clusters}
\begin{restatable}[Mirror Cut]{definition}{mirrorcuts}
    In a graph hierarchy $\tilde G_1, \dots, \tilde G_h$ of graph $G$, for every $j \in [h]$, we define the \emph{mirror cut of $v \in V({\tilde G}_i)$} as follows.
    
    Let $C$ be the cluster of $\tilde G_i$ that includes $v$.
    Then the mirror cut of $v$ is a cut of smallest value among all local cuts that include $v$ in the mirror cluster of $C$ and that has cut value at most $\lmax$.
\end{restatable}

In \Cref{sec:mirrorcuts}, we explain how to efficently maintain local cuts once the cluster hierarchy is computed.
This yields these two results:

(A) First, we have \Cref{alg:buffer}, whose role is to act like a buffer for edge updates.
It takes as an input a (dynamic) mirror cluster, and outputs whether or not it contains a local cut of size strictly less than $\lmin$.
It also outputs a graph that is guaranteed to have no local cut smaller than $\lmin$. 
If the input graph has no such cuts, then the input graph and the output graph are the same.
As we need them in the next section we state the main results about \Cref{alg:buffer} here, they are proven in Section~\ref{sec:mirrorcuts}.

\begin{restatable}[The Buffer Algorithm]{proposition}{buffer}\label{cor:bufferoutput}
    The output graph $G'$ of \Cref{alg:buffer} satisfies that no local cut has value strictly less than $\lmin$.
    Moreover, if at any point in time, the input graph of \Cref{alg:buffer} satisfies that no local cut has value strictly less than $\lmin$, then the output graph and the input graphs are the same. 
    
    If the graph has initial volume $V$ and is subject to updates of total volume\footnote{We define the \emph{volumes} of update as follows: an edge update has volume 1, while a set deletion $S$ has volume $\vol (S)$.} $U$, its preprocessing time is $O\PAR{V \poly(\frac \lmax \phi) \log n}$ and its update time over all updates is $O\PAR{U \poly(\frac \lmax \phi) \log n}$.
    
    The algorithm has failure probability less than $n^{-6}$ after $n^3$ updates.  
\end{restatable}

(B) Second, we also have the Mirror Cuts Algorithm (\Cref{alg:maintainmirrors}). \Cref{alg:maintainmirrors} takes as an input a (dynamic) mirror cluster that is certified to have no local cut strictly smaller than $\lmin$, and maintains its mirror cuts of size smaller than $\lmax$. 

\begin{restatable}[Mirror Cuts Algorithm]{proposition}{maintaincuts}\label{cor:mirrorcutsoutput}
\Cref{alg:maintainmirrors} maintains the mirror cuts of its input graph. Together with \Cref{alg:buffer}, it can either certify that there exists a local cut in the mirror cluster that is strictly smaller than $\lmin$, or it can maintain the mirror cuts of size at most $\lmax$ of the mirror cluster.
    If the graph has initial volume $V$ and is subject to updates of total volume $U$, its preprocessing time is $O\PAR{V \poly(\frac \lmax \phi) \log n}$ and its update time over all updates is $O\PAR{U \poly(\frac \lmax \phi) \log n}$. The algorithm has failure probability less than $n^{-6}$ after $n^3$ updates.  
\end{restatable}

Instances of these algorithms can be efficiently maintained in a specific data structure, where the partition of a graph can only be subjects to splits:

\begin{restatable}{proposition}{runningtimedatastructure}\label{prop:datastructurebufferandcuts}
 Let $\C$ be a (dynamic) partition of $G$, where $G$ has updates of volume $U$ and begins with $m$ edges.
    Suppose that $\C$ starts as $\C=\{V\}$, and can be subject to splits during those updates, that is, an element $C$ in $\C$ can be replaced by $C\setminus S$ and $S$, for some $S\subsetneq C$.
   We can maintain for all clusters instances of the Buffer Algorithm and the Mirror Cuts Algorithm ($\Cref{alg:buffer}$ and \Cref{alg:maintainmirrors}) in $\tilde O((m+U) \poly (\frac \lambda \phi))$ time.
\end{restatable}

\section{Static Algorithm}\label{sec:static}
Before diving into the dynamic algorithm, we present a simple static algorithm that we will dynamize in later sections.
We remind the reader of the \textit{requirements} of this algorithm: It has as input a (sparsified) graph $G$ with minimum cut $\tilde \lambda$, and two fixed values $\lmin$ and $\lmax$.
\begin{enumerate}[noitemsep]
    \item If $\tilde \lambda < \lmin:$ any output is acceptable.
    \item If $\lmin \le \tilde \lambda \le \lmax$, output an approximation of the minimum cut in $G$.
    \item If $\tilde \lambda > \lmax$, output ``$\tilde \lambda > \lmax$'' or a value strictly larger than $\lmax$.
\end{enumerate}

In fact, as the static algorithm will be used as the preprocessing algorithm for the dynamic one, it will also set up the data structure used for the dynamic algorithm thereafter, and thus, build a cluster hierarchy as discussed in the section above.

The main idea is as follows: we first decompose the graph into expanders, then decompose the expanders further into smaller connected graphs, called \emph{clusters}, so that every minimum cut $S$ in the graph can be uncrossed for each cluster that it cuts. In the end, $S$ can be approximated either by a cut that consists of a union of (full) clusters, or by a local cut.
We then contract the clusters and recurse on the resulting graph, as only clustering once reduces the number of edges by an $n^{o(1)}$ factor only.
In fact, to prepare for the dynamic algorithm, the static algorithm will freeze any cluster in which it finds a cut of boundary size less than $\lmin$.

Recall that $\epsilon = \frac 1 {\sqrt{\log n}}$, $\alpha = \frac{1}{\poly \log n}$ and $\phi = 2^{-\Theta(\log^{3/4} n)}$ and that $\epsilon \le 0.04$.
Note that $\frac \alpha \phi \ge 1$. We also set $\gamma = \frac 1 3$.

\begin{algo}[Static Algorithm]\label{alg:StaticAlgorithm}
The static algorithm works as follows:
\begin{enumerate}[label=(\arabic*),noitemsep]
    \item If the graph has 2 nodes, output the number of edges between the two nodes. If the graph has one node, output ``$\tilde \lambda > \lmax$''. In any case, stop this call to the algorithm.

    \item \label{step:static1} Compute an $(\alpha, \phi)$-boundary-linked expander decomposition of $G$, using the algorithm in~\cite{expanderhierarchy}. 

    \item \label{step:static3} Initiate the data structure on the partition output by the expander decomposition. That includes mirror clusters for each of the expanders, and instances of the Buffer Algorithm and the Mirror Cuts Algorithm (\Cref{alg:buffer} and \Cref{alg:maintainmirrors}). Maintain throughout that every cluster that has a local cut strictly smaller than $\lmin$ is frozen.
    
    \item \label{step:static2} Decompose each expander into clusters that have no $(1-\epsilon)$-boundary sparse cuts of cut-size at most ${\lmax}$ and of volume at most $\frac \lmax \phi$ using \Cref{alg:decomposingexpanders} below and maintain or instantiate (for new clusters) an instance of the Buffer Algorithm and the Mirror Cuts Algorithm (\Cref{alg:buffer} and \Cref{alg:maintainmirrors}) for every cluster.
    
    
     \item \label{step:static4} Contract all clusters in $G$ and call \Cref{alg:StaticAlgorithm} recursively on the contracted graph. 

    \item \label{step:static5} 
     If at least one cluster is frozen, report $\Tilde \lambda < \lmin$.
    \item \label{step:static6}
    Else, return the smallest cut value found among the recursive call in Step~\ref{step:static4} and all the calls to the Mirror Cuts Algorithm (\Cref{alg:maintainmirrors}) in Step~\ref{step:static3}, if this value is smaller than $\lmax$.
    Store the pointer  to the cut that achieved this value (if the algorithm is required to return cuts).
    
       If no cluster was frozen and if either no cut is found or the smallest found cut is strictly larger than $\lmax$, then report $\Tilde \lambda > \lmax$. 
\end{enumerate}
\end{algo}

\subsection{Correctness of Algorithm \ref{alg:StaticAlgorithm}}

We next argue that Algorithm~\ref{alg:StaticAlgorithm} fulfills Requirements 1-3.
First, let's argue about the base case.
(A) For a graph with 2 nodes, there is only one cut possible, and that is the cut we return. This trivially satisfies the three requirements of our algorithm.
(B) A graph with only one node has no cut, since every cut needs to be both nonempty and different from $V$. By convention, we say that the size of the minimum cut in a graph with one node is $+\infty$, and, thus, the algorithm reports that $\tilde \lambda > \lmax$.

Let us now argue in the general case.
Requirement 1 is trivially fulfilled.
For Requirement 3 note that if $\Tilde \lambda > \lmax$, the algorithm cannot find a cut smaller than $\lmax$, and will report that $\Tilde \lambda > \lmax$.

It thus remains to show that Requirement 2 is fulfilled, i.e., if $\lmin \le \Tilde \lambda \le \lmax$,  then for any minimum cut $S$ in $G$, we can find a $(1+2\epsilon)$-approximate mincut $S'$  in $G$ with either the calls to the Mirror Cuts Algorithm (\Cref{alg:maintainmirrors}) in Step~\ref{step:static3} or the recursive call of Step~\ref{step:static4}. For this we rely on the following lemma.

\begin{lemma}
    In \Cref{alg:StaticAlgorithm}, the contracted graph of Step~\ref{step:static4} and the mirror clusters of Step~\ref{step:static3} preserve the cut values of the non-contracted cuts of $G$ and given a cut in one of them, one can reconstruct the corresponding cut in $G$.
\end{lemma}

\begin{proof}
    This is immediate as they are contractions of the original graph, and thus preserve cuts.
\end{proof}

As we prove in  \Cref{cor:3.15} in Section~\ref{subsec:clusterdecomposition},
w.h.p.~every unfrozen cluster $C$ contains no $(1-\epsilon)$-boundary-sparse cut $S$ with $\boundary_G(S) \le \lmax$ and $\vol(S) \le \lmax/\phi$  at termination of \Cref{alg:decomposingexpanders} and also w.h.p. an instance of  the Buffer Algorithm and the Mirror Cuts Algorithm (\Cref{alg:buffer} and \Cref{alg:maintainmirrors}) is created or maintained for every cluster in the current cluster decomposition. Thus if there are no frozen clusters then w.h.p.
the graph is decomposed into a valid cluster decomposition.
If a frozen cluster exists, then $\tilde \lambda < \lmin$ and  \Cref{alg:StaticAlgorithm} is allowed to return any value.
As we show in Section~\ref{sec:mirrorcuts}, \Cref{alg:buffer} in Step~\ref{step:static3} correctly finds all local cuts
in the mirror cluster $C$ for every cluster $C$ at this level of the cluster hierarchy. 
Step~\ref{step:static4} calls the algorithm recursively, which ensures the full cluster hierarchy is built until only one or two nodes are left. In the case where $\lmin\le \tilde \lambda \le \lmax$, by \Cref{prop:clusterhierarchy}, the minimum cut can be approximated by a local cut in a mirror graph or in the most-collapsed graph of the cluster hierarchy. Since all such local cuts are computed in Step~\ref{step:static3}, and their minimum is returned in Step~\ref{step:static6}, this completes the correctness proof if $\lmin\le \tilde \lambda \le \lmax$.
In the case where $\tilde \lambda > \lmax$, since we only return a true value of a cut, in Step~\ref{step:static6} we either output a value strictly larger than $\lmax$, or we report that $\tilde \lambda > \lmax$, as the only case where we report that $\tilde \lambda < \lmin$ is in Step~\ref{step:static5}, i.e., when a cluster is frozen, which only happens if a cut smaller than $\lmin$ is found. This concludes correctness, as in the case $\tilde \lambda < \lmin$, any output is acceptable.

We summarize the result as follows.

\begin{lemma}
    Consider the cluster hierarchy output by \Cref{alg:StaticAlgorithm} on graph $G$, and let $S$ be a minimum cut of $G$. 
    With probability at least $1-n^{-5+o(1)}$ there exists a cut $S'$ that is a $(1+2\epsilon)$-approximate mincut such that either:
    \begin{itemize}[noitemsep]
        \item There exists a cluster $C$ and $S'$ is a cut in $G/(V\setminus C)$ of volume at most $2\frac \lmax \phi + \lmax$.
        \item $S'$ is a cut of $G$ and $S'$ crosses no cluster.
    \end{itemize}
\end{lemma}

The probability statements stems from the analysis below where we show that the decomposition on each level is correct with probability at least $1-3n^{-5}$ (\Cref{cor:3.15}) and that we have at most $n^{o(1)}$ levels (\Cref{lem:norecursivecallse}).

\subsection{Decomposing Expanders}\label{subsec:clusterdecomposition}
In this subsection, we present an algorithm, called \Cref{alg:decomposingexpanders}, that implements Step~\ref{step:static2} in Algorithm~\ref{alg:StaticAlgorithm}. It takes as input an expander decomposition, and outputs a decomposition such that 
w.h.p.~every unfrozen cluster $C$ contains no $(1-\epsilon)$-boundary-sparse cut $S$  with $\boundary_G(S) \le \lmax$ and $\vol(S) \le \lmax/\phi$ at termination of \Cref{alg:decomposingexpanders}. It also maintains or instantiates (for new clusters) an instance of the Buffer Algorithm and the Mirror Cuts Algorithm (\Cref{alg:buffer} and \Cref{alg:maintainmirrors}) for every cluster.

Algorithm~\ref{alg:decomposingexpanders} takes as input an expander decomposition, initializes the data structures, and then repeatedly calls Algorithm~\ref{alg:subroutine} on each cluster until no cluster with an unchecked vertex exists. Algorithm~\ref{alg:subroutine} performs the actual cut finding and graph cutting.

Each node stores a boolean value, which represents whether it  is \emph{checked} or \emph{unchecked}. Intuitively, node $v$ being checked means that LocalKCut has already been executed at $v$ and the size of the returned cuts has already been taken into account, while being unchecked means that LocalKCut needs to be executed at $v$. 
Algorithm~\ref{alg:decomposingexpanders} runs as subroutine Algorithm~\ref{alg:subroutine}  that takes a cluster decomposition, a cluster $C$ of that decomposition together with a set of unchecked vertices in $C$ as input and
(a) either finds a cut of size less than $\lmin$ or (b) cuts the cluster along a $(1-\epsilon)$-boundary sparse cut or (c) certifies that no cut is $(1-\epsilon)$-boundary sparse.
Algorithm~\ref{alg:decomposingexpanders} repeats this procedure until no cluster with minimum cut at least $\lmin$ containing a $(1-\epsilon)$-boundary-sparse cut exists. 
Our algorithm will maintain the following invariant:

\cuttingisgood*

The goal of the algorithm is to not miss any boundary-sparse cut. 
We use procedure LocalKCut to find such cuts. \emph{Checking node $v$} refers to executing LocalKCut with $v$ as starting point.
As we will show, it suffices to make sure that for every cut that might be boundary-sparse there is at least one node in the cut on which the subroutine LocalKCut is executed often enough to guarantee that the cut is found with high probability. This is non-trivial as the LocalKCut algorithm only guarantees to find $\gamma$-extreme sets (with constant $\gamma$) of small enough cut value and volume  with high enough probability. As we show, under certain conditions this suffices to guarantee that no cluster contains a $(1-\epsilon)$-boundary sparse cut.
To do so we need to guarantee that each cluster on which LocalKCut is run was first checked for cuts $S$ with $\boundary_G(S) < \lmin$. To check this we use Algorithm~\ref{alg:buffer}. This is actually a dynamic graph algorithm and each cluster $C$ keeps a pointer to its instance of  Algorithm~\ref{alg:buffer}. If $C$ is later split into $S$ and $C \setminus S$ then we simply update the instance for $C$ by edge deletions and insertions to become the instance for the larger of the two sets $S$ and $C \setminus S$. In this way a new instance has to be created for the smaller of the two sets, which is crucial for achieving a near-linear running time.

Note that if a set $S \subseteq C$ has no edges leaving $C$, then $w(S, V \setminus C) = 0 \le w(S, C\setminus S)$ and, thus, $S$ cannot be $(1-\eps)$-boundary-sparse in $C$. Thus, vertices in $C$ that are not incident to a boundary edge of $C$ do not need to be checked in $C$, i.e., we do not need to mark them as unchecked. Hence, only vertices that are incident to a boundary edge are marked as unchecked, all others are marked as checked.

Recall that a cluster becomes frozen if the algorithm finds a cut of size less than $\lmin$. As we will show if there are no frozen clusters when \Cref{alg:decomposingexpanders} terminates
then the maintained decomposition of the vertices into clusters is with high probability such that no $(1-\epsilon)$-boundary sparse cut of size at most $\lmax$ and volume at most $\frac \lmax \phi$ exists.

We call the decomposition of the vertex set after the algorithm terminates the \emph{cluster decomposition}.

\begin{algo}[Decomposing Expanders]\label{alg:decomposingexpanders}
    Input: an $(\alpha, \phi)$-boundary-linked expander decomposition $V_1, \dots, V_\ell$.

\begin{enumerate}[noitemsep]
    \item \label{step:de2}  We mark every node incident to an inter-expander edge as unchecked, and all other nodes as checked.
\item \label{step:de3}  We set the cluster decomposition to $\C=\{ V_1, \dots, V_\ell\}$.
\item \label{step:de4} For each cluster $C$, we 
set $C$ as unfrozen if  \Cref{alg:buffer} reports $C$ has no local cut of size strictly less than $\lmin$, and as frozen otherwise. In any case, we keep at $C$ a pointer to its instance of \Cref{alg:buffer}.
\item While there exists in $\C$ an unfrozen cluster $C$ with an unchecked vertex, we run the Find and Cut subroutine (\Cref{alg:subroutine}) on $C$, which might change $\C$.
\end{enumerate}
\end{algo}

\begin{algo}[Find and Cut Subroutine]\label{alg:subroutine}
Input: a graph $G$, a cluster decomposition $\mathcal C$ of $G$, and an unfrozen cluster $C$ from that decomposition $\mathcal C$, together with a set of unchecked vertices satisfying \Cref{inv:cuttingisgood} in $G$ and for each cluster a pointer to its instances of $\Cref{alg:buffer}$ and \Cref{alg:maintainmirrors}.

    \item\label{step:fc2}  For every unchecked vertex $v$ in $C$:
\begin{enumerate}[noitemsep]
    \item \label{step:fcl1} Execute  $10\log n \PAR{\frac \lmax \phi}^6 \lmax^4$ many calls to LocalKCut  with parameters $G=C_v, v, \nu = \frac {\lmax} \phi, {\lmax}$, where $C_v$ is the current cluster to which $v$ belongs. Let $\mathcal{S}$ be the set of outputs from all LocalKCuts.
    \item For every $S \in \mathcal{S}$:

    \label{step:fcl2b}
    Check if $S$ is $(1-\epsilon)$-boundary-sparse.
    If it is the case, cut the cluster along $S$, thereby removing $C$ and creating two new clusters $S$ and $C \setminus S$, and mark any node incident to a newly cut edge as unchecked. The node $v$ remains unchecked. Create a new instance of the Buffer Algorithm and the Mirror Cuts Algorithm (\Cref{alg:buffer} and \Cref{alg:maintainmirrors}) for $S$ (or $C\setminus S$ if it has a smaller volume than $S$), and update the instance of the Buffer Algorithm and the Mirror Cuts Algorithm (\Cref{alg:buffer} and \Cref{alg:maintainmirrors}) of $C$ so that it no longer runs on the graph induced by $C$ but instead on the graph induced by $C\setminus S$ (respectively $S)$, using a set deletion $S$ and $O(\vol (S))$ edge insertions. Freeze any cluster as necessary (based on the output of \Cref{alg:buffer}). End the current call to \Cref{alg:subroutine}.
    \item \label{step:fcl3} 
    Mark the node $v$ as checked and move on to the next unchecked vertex in $C$.

\end{enumerate}
\end{algo}

We call the set of $10\log n \PAR{\frac \lmax \phi}^6 \lmax^4$ many calls to LocalKCut in Step~\ref{step:fcl1} a \emph{batch of LocalKCut calls}. A batch of LocalKCut calls is \emph{correct} or they \emph{do not fail} if there exists one (or many) $(1-\epsilon)$-boundary-sparse cuts of volume at most $\frac \lmax \phi$ and size at most $\lmax$, and at least one run of LocalKCut in the batch finds one of them.
\subsubsection{Correctness}

We need to show that 
w.h.p. Invariant~\ref{inv:cuttingisgood} holds at the termination of \Cref{alg:decomposingexpanders}.
For that, we start by showing that the invariant holds before the while loop, and that each step of Algorithm~\ref{alg:subroutine} maintains   Invariant~\ref{inv:cuttingisgood}. 
Thus, the invariant holds after Algorithm~\ref{alg:decomposingexpanders} has terminated. As there are no unchecked nodes left in unfrozen clusters at that point in time,  it follows that
w.h.p. every unfrozen cluster $C$ contains no $(1-\epsilon)$-boundary-sparse cut $S$  with $\boundary_G(S) \le \lmax$ and $\vol(S) \le \lmax/\phi$.

\begin{lemma}\label{lem:3.12}
    Right after Step~\ref{step:de4} in \Cref{alg:decomposingexpanders},  \Cref{inv:cuttingisgood} holds.
\end{lemma}

\begin{proof}
    Right  after Step~\ref{step:de4} every cluster is an expander in the expander decomposition. Thus,
    any cut $S$ in a cluster $C$  can only be $(1-\epsilon)$-boundary-sparse if it contains a node $v$ that is incident to an inter-expander edge, that is, if $w(S, V\setminus V_i) \neq 0$ as we have $(1-\epsilon) w(S, V\setminus V_i) > w(S, V_i\setminus S) \ge 0$, by the definition of boundary-sparsity.
     As all the nodes incident to all inter-expander edges were marked as unchecked in Step~\ref{step:de2}, \Cref{inv:cuttingisgood} holds.
\end{proof}

\begin{lemma}\label{lem:3.13}
    Assuming that \Cref{alg:buffer} and each batch of LocalKCut calls do not fail, each execution of~\Cref{alg:subroutine}
    maintains \Cref{inv:cuttingisgood}.
\end{lemma}

\begin{proof}

\textbf{Step~\ref{step:fcl1}:}
Step~\ref{step:fcl1} of \Cref{alg:subroutine} neither affects the marking of vertices nor the cluster decomposition and, thus, does not affect \Cref{inv:cuttingisgood}.

\textbf{Step~\ref{step:fcl2b}:}
We will next show that the invariant is maintained in Step~\ref{step:fcl2b}. 
In Step~\ref{step:fcl2b} 
an unfrozen cluster $C$ is cut along a  $(1-\epsilon)$-boundary-sparse cut $C' \subsetneq C$ into a new cluster $C'$ and into $C \setminus C'$ and the endpoints of new inter-cluster edges are marked as unchecked.
To show \Cref{inv:cuttingisgood} we need to show that 
every $(1-\epsilon)$-boundary sparse cut $S$ in $C'$  with 
 $\vol(S) \le \frac \lmax \phi$ and $w(S, C'\setminus S) \le \lmax$ contains an unchecked vertex.
 
    Assume by contradiction that no vertex in $S$ is unchecked
    at the end of Step~\ref{step:fcl2b}, i.e. every vertex in $S$ is checked. As cutting does not change the marking of the vertices, this implies that (a) every vertex in $S$ was also checked in $C$, and (b) no edge incident  to $S$ was cut, i.e., became a new inter-cluster edge. Specifically, there exists no edge from a node in $S$ to a node in $C\setminus C'$, as the algorithm would have marked any endpoint in $S$ of such an edge as unchecked. 
    To achieve a contradiction we will show that $S$ must have been (i) a $(1-\epsilon)$-boundary-sparse cut in $C$, (ii) that $S$ has volume at most $\lmax/\phi$, and (iii)  $\lmin \le w(S, C\setminus S) \le \lmax$. As, by assumption, \Cref{inv:cuttingisgood} was satisfied in $C$ before Step~\ref{step:fcl2b}, it follows that $S$ must have contained an unchecked vertex in $C$. Since no vertex is checked in Step~\ref{step:fcl2b}
    it follows that $S$ still must contain an unchecked vertex, leading to a contradiction.
    
    It remains to show (i) to (iii).
    We first show (i). Recall that $C' \subset C$.
    Then we know that $w(S, C'\setminus S) = w(S, C\setminus S)$, as otherwise there would be an edge from a node in $S$ to a node in $C\setminus C'$. Thus $w(S, C \setminus S) = w(S, C' \setminus S)$ and $w(S, V\setminus C') = w(S, V\setminus C)$.
    Since $C'$ is a $(1-\epsilon)$-boundary-sparse cut in $C$, we have that $w(C', C\setminus C') \le (1-\epsilon) w(C\setminus C', V\setminus C) < w(C\setminus C', V\setminus C)$.
    Therefore, 
    \begin{align*}
        w(S, C\setminus S) = w(S, C'\setminus S)&\le(1-\epsilon) w(C'\setminus S, V\setminus C')\\
        &=(1-\epsilon)\PAR{w(C'\setminus S, V\setminus C) + w(C'\setminus S, C\setminus C')}\\
        &\le(1-\epsilon)\PAR{w(C'\setminus S, V\setminus C) + w(C', C\setminus C') } \\
        &< (1-\epsilon)\PAR{w(C'\setminus S, V\setminus C) + w(C\setminus C', V\setminus C)}\\
        &=(1-\epsilon)w(C\setminus S, V\setminus C) 
    \end{align*}
    Furthermore,
    \begin{align*}
        w(S, C\setminus S) = w(S, C'\setminus S)&\le(1-\epsilon) w(S, V\setminus C') = w(S, V\setminus C)
    \end{align*}
    Thus $S$ is $(1-\epsilon)$-boundary sparse in $C$.

    We next show (ii). As $S \subset C' \subset C$ it follows that the volume of $S$ in $C$ equals the volume of $S$ in $C'$, and thus, it is at most $\lmax/\phi$. 

    Finally we show (iii). As there exists no edge from a node in $S$ to a node in $C \setminus C'$ it follows that $w(S, C'\setminus S) = w(S, C \setminus S)$ and, thus,
    $w(S, C\setminus S) =w(S, C'\setminus S)\le \lmax$.

    Thus (i) to (iii) hold, completing the proof for Step~\ref{step:fcl2b}.

\textbf{Step~\ref{step:fcl3}:}
It remains to show that the invariant is maintained in Step~\ref{step:fcl3}. If this step is executed, (a) the status of only one node can change, namely the one of $v$ and (b) the execution of the for-loop for $v$ modified neither the cluster decomposition nor the status of any other unchecked node. Thus~\Cref{inv:cuttingisgood} continues to hold for all clusters that do not contain $v$. We need to argue that it also holds for the (unique) cluster $C_v$ that contains $v$.
For that, we will start by making the following claim, which is crucial to ensure that LocalKCut finds all the cuts necessary for our analysis:

\begin{claim}\label{claim:4.1}
    For any cluster $C$ that is unfrozen (after the correct execution of \Cref{alg:buffer} on it) and any cut $S \subseteq C$, if $\vol(S)\le \frac \lmax \phi$, then $S$ contains no cut $S' \subseteq S$ with $\boundary_GS' < \lmin$ and $S$ contains no cut $S' \subseteq S$ with $\boundary_G(C\setminus S') < \lmin$.
\end{claim}

\begin{proof}
Assume by contradiction that such an $S'$ exists and consider two cases.
\begin{enumerate}[noitemsep]
    \item In the case $\boundary_G S'<\lmin$, we have that $\vol(S')\le \vol (S) \le \frac \lmax \phi$, and thus $S'$ is a local cut in the mirror cluster. By~\Cref{cor:bufferoutput}, \Cref{alg:buffer} on cluster $C$  would have found $S'$ w.h.p. and, thus, marked $C$ as frozen, leading to a contradiction to $C$ being unfrozen.
    \item In the case $\boundary_G(C\setminus S') < \lmin$, let $X$ be the cut in the mirror cluster of $C$, composed of $S'$ and the vertex $u$ representing the contraction of $G\setminus C$. 
    Then $\vol(u) = \boundary_G(C) = w(v, C\setminus S) + w(v, S) \le \boundary_G (C\setminus S)+\vol(S) \le \lmin + \frac \lmax \phi$. This fact combined with the fact that $\vol(S) \le \lmax/\phi$ implies that $\vol (X)  = \vol(S') + \vol(u) \le \vol(S) + \vol(u) \le \lmin + 2\frac \lmax \phi$, and thus that $X$ is a local cut. 
    Furthermore, $\boundary_G(X) = \boundary_G(C \setminus S') < \lmin$ By~\Cref{cor:bufferoutput} \Cref{alg:buffer} on cluster $C$  would have found $C \setminus S'$ w.h.p. and, thus, marked $C$ as frozen, leading to a contradiction to $C$ being unfrozen.
\end{enumerate}

This concludes the proof of the claim.
\end{proof}
We now continue arguing about Step~\ref{step:fcl3}.
Assume by contradiction that a cut $S$ exists in  cluster $C_v$ at the end of Step~\ref{step:fcl3}
such that $S$
(a) is $(1-\epsilon)$-boundary sparse in $C_v$,
(b) has volume at most $\lmax/\phi$,
(c) also satisfies $w(S, C_v\setminus S) \le \lmax$, and (d)
 contains only checked vertices. As the invariant held at the beginning of 
Step~\ref{step:fcl3}, $S$ must have had an unchecked vertex then and as $v$ is the only vertex whose status changed since then, $v$ must belong to $S$.
As the cluster decomposition did not change in the execution of the for-loop of $v$, it follows that $C_v$ is the cluster on which this call of \Cref{alg:subroutine} is executed. Thus $C_v$ is unfrozen and the LocalKCut subroutines in Step~\ref{step:fcl1} were executed on $C_v$.

Let $S^*$ be the cut with smallest cardinality $|\tilde S|$ among all cuts $\tilde S$ with
$v \in \tilde S$, 
$\tilde S \subseteq S$, $w(\tilde S,C_v \setminus \tilde S) \le \lmax$, and $\tilde S$ is
$(1-\epsilon)$-boundary sparse in $C_v$.
Such a cut $S^*$ must exist as $S$ could be this cut.
As $C_v$ is unfrozen, we can apply Claim~\ref{claim:4.1} to it.
It follows that every strict subset $S'$ of $S^*$ with $v \in S'$, $w(S', C_v \setminus S') \le \lmax$ is not $(1-\epsilon)$-boundary sparse in $C_v$ and $\boundary_G(S') \ge \lmin$ \emph{and} $\boundary_G(C_v \setminus S') \ge \lmin$.

Now Claim~\ref{lem:3.14} for the unfrozen cluster $C_v$  shows that $S^*$ must be $1/3$-extreme. This allows us now to apply 
 Lemma~\ref{lem:localkcut}. It shows that
LocalKCut started on $v$ returns a $1/3$-extreme cut $S^*$ containing $v$ with probability $\Omega(|S^*|^{-6}\lmin^{-4}) = \Omega(|S^*|^{-6} \lmax^{-4}) $. As the number of calls is $\Theta(\log n (\lmax/\phi)^6 \lmax^4)$ and $|S^*| \le |S| \le \lmax/\phi$, it follows $S^*$ would have been returned by LocalKCut with high probability $p$, whenever LocalKCut is started at $v$, i.e., in the current iteration of the for-loop of \Cref{alg:subroutine}. The algorithm, however, reached Step~\ref{step:fcl3} without finding $S^*$, as $S$, and, thus, $S^*$ has not been cut from $C'$. This can only happen with  probability $1-p$. It follows that with probability $p$, no such cut $S^*$ and, thus, no cut $S$ can exist in $C'$.

The probability of $1-n^{-8}$ stems from the fact that the only randomized procedure is the batch of calls to LocalKCut, which succeeds with probability at least $1-n^{-10}$ by \Cref{lem:chernoff}.
Since we call it at most $n^2$ times, as it is run from each vertex at most $n$ times, once each time the cluster it is in gets cut into two more, the result holds.
\end{proof}
It remains to prove \Cref{lem:3.14}.

\begin{claim}\label{lem:3.14}
Assume we have an unfrozen cluster $C$ on which the output of \Cref{alg:buffer} is correct.
    If a cut $S \subsetneq C$ satisfies that:
    \begin{enumerate}[noitemsep]
        \item  $w(S, C\setminus S) \le \lmax$, 
        \item every strict subset $S'\subsetneq S$ with $w(S', C\setminus S') \le \lmax$  is \emph{not} a $(1-\epsilon)$-boundary-sparse cut in $C$,  has $\boundary_GS' \ge \lmin$, and $\boundary_G(C \setminus S') \ge \lmin$
    \end{enumerate}
    then $S$ is $\frac 1 3$-extreme.
\end{claim}

\begin{proof}

    By contradiction assume there exists a cut $S$ that fulfills the conditions of the claim, but is not $\frac{1}{3}$-extreme.
    We will find a cut $S' \subset S$ that either is $(1-\epsilon)$-boundary-sparse or satisfies $\boundary_GS' < \lmin$, leading to a contradiction that $S$ fulfills the conditions of the claim. 
    
    Since $S$ is not $\frac{1}{3}$-extreme, it contains as strict subset a cut $S'$ that satisfies $w(S', C\setminus S')\le \frac {w(S, C\setminus S)} 3 \le \frac \lmax 3 \le 0.4\lmin$. Thus $S'$ is a strict subset of $S$ with
    $w(S', C\setminus S') \le \lmax$.
    It follows that  from the conditions of the claim that $\boundary_GS' \ge \lmin$ and $\boundary_G(C \setminus S') \ge \lmin$. Thus, $w(S', V\setminus C) \ge \boundary S' - w(S', C\setminus S') \ge 0.6\lmin$.
    Similarly, $w(C\setminus S', V\setminus C) \ge \boundary (C \setminus S') - w(S', C\setminus S') \ge 0.6\lmin$, and thus $S'$ is $(1-\epsilon)$-boundary-sparse in $C$.
\end{proof}

\Cref{lem:3.12} and \Cref{lem:3.13} imply that \Cref{inv:cuttingisgood} holds at the termination of \Cref{alg:decomposingexpanders}.

\begin{lemma}\label{cor:3.15}
 With probability at least $1-3n^{-5}$ it holds that
 (a) at the termination of \Cref{alg:decomposingexpanders} \Cref{inv:cuttingisgood} holds; 
 (b) the cluster decomposition output by \Cref{alg:decomposingexpanders} satisfies that no unfrozen cluster contains a $(1-\epsilon)$-boundary-sparse cut of size at most $\lmax$ and volume at most $\frac \lmax \phi$;
 (c) at every cluster that exists in the final cluster decomposition an instance of the Buffer Algorithm and the Mirror Cuts Algorithm (\Cref{alg:buffer} and \Cref{alg:maintainmirrors}) is stored at the cluster.
\end{lemma}

\begin{proof}
Let us first assume that all outputs by \Cref{alg:buffer} and the batches of LocalKCut are correct.

    Since we only uncheck a vertex if it is contained in no cuts that are $(1-\epsilon)$-boundary-sparse with cut value at most $\lmax$ and volume at most $\frac \lmax \phi$, this lemma relies on \Cref{inv:cuttingisgood}.
    By \Cref{lem:3.12}, the invariant holds at the beginning of the while loop in \Cref{alg:decomposingexpanders}.

     We now have to show that the invariant holds until the end of the while loop.
     This is now immediate as the while loop is simply a call to \Cref{alg:subroutine}, and \Cref{alg:subroutine} does not affect the invariant by \Cref{lem:3.13}.
     
       The fact that the while loops ends when each cluster either is frozen or contains no unchecked vertex, together with \Cref{inv:cuttingisgood}, conclude the proof in the case where all outputs by \Cref{alg:buffer} and the batches of LocalKCut are correct.

    Let us now analyze the probability of failure.
    Each time we call a batch of LocalKCut, we ensure correctness with probability at least $1-n^{-10}$, by \Cref{lem:chernoff}. 
    Let us now look at each vertex $v$.
    We can run a batch of LocalKCut with starting node $v$ at most once per while loop of \Cref{alg:decomposingexpanders}. As each vertex belongs to exactly one cluster at any point in time, it follows that $v$ can only be the starting node for a batch of LocalKCut calls if
    the cluster of $v$ has been split.
    Since each cluster can be split at most $n$ times, each vertex can be the starting node for at most $n$ batches of LocalKCut calls. 
    Thus, there can be at most $n^2$ batches of LocalKCut calls in \Cref{alg:decomposingexpanders}. 
    As the probability that a given  batches of LocalKCut calls fails is at most $n^{-10}$, all batches are correct with probability at least $1-n^{-8}$.

Let us now upper bound the probability that at least one instance of \Cref{alg:buffer} is wrong. 
Note that there are at most $n$ instances of \Cref{alg:buffer} created, as these instances never terminate and at the end of \Cref{alg:decomposingexpanders} each instance of \Cref{alg:buffer} still contains at least one vertex. Also note that there are at most $n^{1+o(1)}$ updates in each instance, as it is updated only when the cluster it is run on is split. 
The cluster can be split at most $n$ times, and each time we split a set $S$ that was found by a run of LocalKCut, and thus has volume at most $\frac \lmax \phi=n^{o(1)}$.
As there are at most $n^2$ edges in an instance of \Cref{alg:buffer} at initialization, it follows that there can be at most $n^2 + n^{1+o(1)} \le 2n^2$ updates in an instance of \Cref{alg:buffer} over all the executions of Step \ref{step:fcl3}. Thus there are $2n^2$ update operations in any instance. Hence
 \Cref{cor:bufferoutput} shows that each instance fails with probability at most $n^{-6}$ during these $2n^2$ updates.
 A similar analysis for the Mirror Cuts Algorithm (\Cref{alg:maintainmirrors}) using \Cref{cor:mirrorcutsoutput} gives the same result.
Since there are at most $n$ clusters, by a union-bound, this concludes the proof.  
\end{proof}

\subsubsection{Bounding the number of inter-cluster edges}
This subsection is dedicated to the analysis of the number of inter-cluster edges in our final cluster decomposition, i.e., it proves \Cref{lem:numedges}.

\begin{proposition}
    \label{lem:numedges}
    Let $M$ be the number of inter-expander edges after the expander decomposition in \Cref{alg:StaticAlgorithm}.
    Then there are $M\frac {2^{\frac {O(1)}{\epsilon}}} {\epsilon}$ many inter-cluster edges in the cluster decomposition after the execution of \Cref{alg:decomposingexpanders}.
\end{proposition}


    Consider the following rooted forest $F$ in which each node represents a cluster of $G$ at \emph{some point} during the computation: the roots are the expanders output by the expander decomposition.
    For each node representing a cluster, its two children (if any) are the clusters formed by splitting the corresponding cluster. By slight abuse of notation we do not distinguish between clusters of the cluster decomposition and nodes in the forest.

    We first show that the size of the boundary of a node is always smaller than the one of its parent by at least $\epsilon {\lmin}/2$.

\begin{lemma}\label{lem:boundariessmaller}
    Let $C$ be a node in $F$ and let $S, C\setminus S$ be its children.
    Then we have that $\boundary C \ge \max\{\boundary S, \boundary (C\setminus S)\}+(\epsilon \lambda_{\min})/2$.
\end{lemma}

\begin{proof}
    Since $S$ is a $(1-\epsilon)$-boundary-sparse cut, we have that $$w(S, C\setminus S) \le (1-\epsilon)\min\{ w(S, V \setminus C), w(C\setminus S, V\setminus C)\},$$
    which implies that $w( C\setminus S, V\setminus C) > w( C\setminus S, S)$.
    As it holds that $w( C\setminus S, V\setminus C) + w( C\setminus S, S) = \boundary (C\setminus S) \ge \lambda_{\min}$, it follows that $w( C\setminus S, V\setminus C) \ge \frac {\lambda_{\min}} 2 $.
    
    The boundary-sparseness of $S$ also implies that $$\boundary S = w(S, V\setminus C) + w(S, C\setminus S) \le w(S, V\setminus C) + (1-\epsilon) w( C\setminus S, V\setminus C) = \boundary C - \epsilon w( C\setminus S, V\setminus C)$$ i.e., $$\boundary C \ge \boundary S +  \epsilon w( C\setminus S, V\setminus C) \ge \boundary S + \epsilon \lmin/2.$$


    The same bound holds by symmetry for $\boundary (C\setminus S)$ and the result follows.    
\end{proof}
Nonte that the lemma implies that $\boundary C > \boundary S$ and $\boundary(C \setminus S)$.
Thus, along any path in $F$ from a root to a leaf the boundary size of the clusters strictly decreases. Specifically, if a cluster has boundary size less than some value, let's say $3\lmax$, then all its descendants in $F$ has boundary size less than $3 \lmax$.
We now look only at the subforest $F'$ of $F$ created from $F$ by restricting $F$ to the clusters with boundary size at least $3\lmax$. More formally, $F'$ is defined as follows: All roots of $F$ are also in $F'$, and moreover, for every $C\in F'$ with $\boundary C \ge 3\lmax$, all its children in $F$ are also in $F'$. However, for $C \in F'$ with $\boundary C < 3 \lmax$, their children do not belong to $F'$, i.e., $C$ is a leaf in $F'$.
As discussed, $F \setminus F'$ contains no node of boundary larger than $3\lmax$.

\begin{lemma}
    The total number of inter-cluster edges of the leaves of $F'$ is at most $O(\frac M \epsilon)$, where $M$ is the number of inter-expander edges output by the expander decomposition.
\end{lemma}

\begin{proof}
We count the number of inter-cluster edges by bounding the sum of sizes of the boundaries for every cluster that is a descendant of a root cluster $V_i$. 
Any inter-cluster edge that is created at some descendant of $V_i$  is also an inter-cluster edge at a leaf descendant of $V_i$ in $F'$ since clusters are never merged. Thus, it suffices to bound the sum of the boundary sizes of the leaf descendants of $V_i$ in $F'$. Let $\C_f$  contain exactly these clusters.
We bound this sum by splitting $\C_f$ into two sums that we will analyze separately:
 \begin{align*}
        \sum_{C \in \C_f} \boundary C &= \sum_{\substack{C \in \C_f, \boundary C \le 2.2\lmax}} \boundary C+\sum_{\substack{C \in \C_f, \boundary C > 2.2\lmax}} \boundary C
        \intertext{and we define} S_1 :&= \sum_{\substack{C \in \C_f, \boundary C \le 2.2\lmax}} \boundary C \intertext{ and } S_2 :&= \sum_{\substack{C \in \C_f, \boundary C > 2.2\lmax}} \boundary C.
\end{align*}
To proceed we  will analyze the potential function $\Phi(\C) := \sum_{C \in \C} \Phi(C)$ were $\Phi(C) = \max\{ 0 , \boundary C - 2.1 \lmax\}$. 
Right after the computation of the expander decomposition the cluster decomposition of $V_i$ only consists of  $V_i$ itself (that is, an expander from the expander decomposition). While computing the cluster decomposition, the algorithm might partition a descendant  $C$ of $V_i$, replacing it by two new clusters $C_1$ and $C_2$ and making $C_1$ and $C_2$ its children in $F'$. We show next that the sum of the potentials of the two children clusters is at least $(\epsilon \lmin)/2$ smaller than the potential of $C$.
To do so we analyze three cases: 
    \begin{enumerate}
        \item If both $C_1$ and $C_2$ satisfy $\boundary C_1, \boundary C_2 \ge 2.1\lmax$, then the total potential changes by:
        \begin{multline*}
        \Phi(C_1) + \Phi (C_2) - \Phi (C) = \boundary C_1 - 2.1 \lmax + \boundary C_2 - 2.1 \lmax -\boundary C + 2.1\lmax \\= -2.1 \lmax + 2 w(C_1, C_2) \le -2.1 \lmax + 2 \lmax \le -0.1 \lmax
        \end{multline*}

        \item If wlog $C_1$ and $C_2$ satisfy $\boundary C_1\le 2.1\lmax \le\boundary C_2 $, then the total potential changes by:
        \begin{multline*}
        \Phi(C_1) + \Phi (C_2) - \Phi (C) = 0 + \boundary C_2 - 2.1 \lmax -\boundary C + 2.1\lmax = \boundary C_2 - \boundary C \le -(\epsilon \lambda_{\min})/2
        \end{multline*}
        where the last inequality follows from \Cref{lem:boundariessmaller}.
        \item If both $C_1$ and $C_2$ satisfy $\boundary C_1, \boundary C_2 \le 2.1\lmax$, then the total potential changes by:
        $$
        \Phi(C_1) + \Phi (C_2) - \Phi (C) = 0 + 0 -\boundary C + 2.1\lmax \le -0.9 \lmax
        $$
    \end{enumerate}

    It follows that the potential drops by at least $(\epsilon \lambda _{\min})/2$ each time when a cluster is split into two smaller clusters. 
    
    We now bound the number of clusters in $\C_f$ as follows: Since the potential drops by $\epsilon \lambda_{\min}/2$ each time a cluster is split, and splitting a cluster only increases the number of clusters by one, we know that the total number of clusters in $\C_f$ is at most $\frac {2\Phi(\{V_i\})} {\epsilon \lmin} \le \frac{2\boundary V_i}{\epsilon \lmin}$.
    As every cluster contributing to $S_1$ contributes at most $2.2\lmax$, it follows that $S_1$ is at most $\lmax \frac {4.4 \boundary V_i} {\epsilon \lmin} = O\PAR{\frac{\boundary V_i \lmax}{\epsilon \lmin}}.)$
   
    For the second sum, $S_2$, note that if $\boundary C \ge 2.2\lmax$, then $\boundary C = O(\Phi(C))$. 
    Recall that the sum of potentials of the two children of a node in $F'$ is less than the potential of the parent. It follows by induction that for any set of descendants of $V_i$ such that none is an ancestor of the other, the sum of their potentials is less than the potential of $V_i$. As the clusters in $C_f$ fulfill this condition, it follows that
    $S_2 = \sum_{\substack{C \in \C_f\\\boundary C > 2.2\lmax}} O(\Phi( C)) = O(\Phi(\{V_i\})) = O\PAR{\frac{\boundary V_i \lmax}{\epsilon \lmin}}$.    

    Summing over all expanders, and using that $\lmax = O(\lmin)$, we get the result.
\end{proof}

We now look at the forest $F''=(F\setminus F')\union \mathrm{leaves}(F')$. 
The leaves of $F'$ are the roots of $F''$.
By the lemma above, we know that the sum of the boundaries of the roots of $F''$ is at most $O\PAR{\frac M \epsilon}$.
We moreover know that any cluster $C$ with children in $F''$ satisfies $\boundary C < 3\lmax$.

\begin{lemma}\label{lem:inter-clusteredges}
    The sum of the boundaries of the leaves of $F''$ is at most $\frac M \epsilon {2^{O(\frac {1}{\epsilon})}}$, where $M$ is the number of inter-expander edges output by the expander decomposition.
\end{lemma}

\begin{proof}
    Let us look at a root $C$ with children in $F''$, and all its descendants.
    We have that $\boundary C < 3 \lmax$, and that the boundary of a node is always $\Omega (\epsilon \lmin)$ smaller than its parent. Thus the depth of any subtree of $F$ that is rooted at a leaf of $F'$ is at most $\frac {3 \lmax}{\Omega(\epsilon \lmin)} = O(1/\epsilon)$.

    Therefore, the number of nodes in each tree is at most $2^{\frac {3 \lmax}{\Omega(\epsilon \lmin)}} \le 2^{O(\frac {1}{\epsilon})}$.
    Hence, in the tree rooted at a leaf $C$ of $F'$, the sum of the boundaries is at most $\boundary C  2^{O(\frac {1}{\epsilon})}$.
    Summing over all trees in $F''$, we get the result.
\end{proof}

\subsubsection{Running Time Analysis}
Now we are ready to analyze the running time of the static algorithm.
\begin{definition}
    Let $u$ be a vertex. 
    We say that $u$ is \emph{responsible} for a cluster $S$ if, in \Cref{alg:subroutine}, $S$ was cut from its parent cluster after being found by LocalKCut with starting vertex $u$.
\end{definition}

\begin{lemma}\label{lem:nottoomanykargers}
    Let $u$ be a vertex. Vertex $u$ is responsible for at most $O(\frac 1 \phi)$ many clusters in $F$.
\end{lemma} 

\begin{proof}
    The first cut $u$ is responsible for will reduce the volume of the cluster $u$ is in to at most $\frac \lmax \phi$.
    Then every further cut will reduce this volume by at least $\lmin$ as it cuts a at least $\lmin$ edges out of the cluster.
    As $\lmax/\lmin = O(1)$ the lemma follows.
\end{proof}

This lemma is important, as after checking a node, it is not guaranteed to become unchecked, as if a boundary-sparse cut of size at least $\lmin$ is found, a cut is made, but the node stays unchecked. 
Hence this bounds how many times a vertex is being checked, that is, how many times we need to run LocalKCut on any individual vertex. It follows that we can bound the total number of call to LocalKCut.

\begin{lemma}\label{lem:nooflocalkcuts}
    The total number of calls to LocalKCut in \Cref{alg:decomposingexpanders} is $\Tilde O\PAR{M\frac {2^{\frac {O(1)}{\epsilon}}} {\epsilon} \PAR{\frac {\lmax^{10}} {\phi^7}}}$, where $M=O(m \phi)$ is the number of inter-expander edges.
\end{lemma}

\begin{proof}
    A vertex only becomes unchecked when it is incident to an inter-cluster edge. By \Cref{lem:inter-clusteredges} we have $O\PAR{M\frac {2^{\frac {O(1)}{\epsilon}}} {\epsilon}}$ inter-cluster edges, and, thus, this bounds the total number of times that a vertex becomes unchecked. From each unchecked vertex $u$
    we run  $O\PAR{\frac {\lmax^{10}} {\phi^6} \log n }$ calls to LocalKCut from  $u$ in \Cref{step:fcl1} of \Cref{alg:subroutine}. 
    As \Cref{lem:nottoomanykargers} shows, after $O(\frac 1 \phi)$ such executions of \Cref{step:fcl1} vertex $u$ must become checked.
    Thus the total number of calls to LocalKCut is as stated in the lemma.
\end{proof}

We remind the reader of the following constants: $\epsilon = \frac 1 {\sqrt{\log n}} \le 0.04.$ $ \lmin=(54 (1-\eps) \ln n )/\eps^2$ and $\lmax = (54 \cdot 1.1 (1+\eps) \ln n)/\eps^2$, $\phi = 2^{-\Theta (\log^{3/4}n)}$ and $ \alpha=\frac 1 {\poly \log n}.$

\begin{lemma}\label{lem:norecursivecallse}
    The number of recursive calls of \Cref{alg:StaticAlgorithm} is $n^{o(1)}$.
\end{lemma}

\begin{proof}
Let $A$ be the number of edges input to the $\ell$-th recursive call.
Step~\ref{step:static1} of the algorithm creates $M=\tilde O(A\phi)$ many interexpander edges, as proven in~\cite{expanderhierarchy}.
    \Cref{lem:inter-clusteredges} shows that then the number of intercluster edges, which is the number of edges in the input graph of the $\ell+1$-th recursive call is $A \phi \cdot \epsilon {2^{\frac {O(1)}{\epsilon}} \le \frac A {2^{a\log ^{3/4}n}}}$ for some $a \in \R_+$. 
Hence, the number of edges at the $r$-th level of the cluster decomposition is at most $\frac m {2^{ar \log ^{3/4}n}}$.
This shows that after $r = \frac 1 a \log^{1/4}n $ levels, the number of edges drops to less than $2$.
\end{proof}

\begin{theorem}
    The total running time of \Cref{alg:StaticAlgorithm} is $O(m^{1+o(1)})$.
\end{theorem}

\begin{proof}
    Let us first analyze one level of the recursion: Step~\ref{step:static1} takes $O(m^{1+o(1)})$ as proven in~\cite{expanderhierarchy}.
    By \Cref{prop:datastructurebufferandcuts}, maintaining the instances of the Buffer Algorithm and the Mirror Cuts Algorithm (\Cref{alg:buffer} and \Cref{alg:maintainmirrors}) takes $O(m^{1+o(1)})$ time in Step~\ref{step:static3}.
    It creates $M=m \phi$ inter-expander edges.
    By \Cref{lem:nooflocalkcuts}, Step~\ref{step:static2} runs  LocalKCut a total of $O\PAR{M\frac {2^{\frac {O(1)}{\epsilon}}} {\epsilon} \PAR{\frac {\lmax^{10}} {\phi^7}}}$ times, each instance taking $n^{o(1)}$, for a total of $O(m^{1+o(1)})$ time overall.
    
    Steps~\ref{step:static4},~\ref{step:static5} and~\ref{step:static6} take $O(m)$ time (excluding the recursive call).

By \Cref{lem:structuretime}, the rest of the data structure can be maintained in $O(m^{1+o(1)})$.

   By \Cref{lem:norecursivecallse}, the number of levels is $n^{o(1)}$. Hence the claim follows.
\end{proof}

\section{The Dynamic Algorithm}\label{sec:dynamic}
The dynamic algorithm works as follows: at the beginning, the algorithm runs the static algorithm, which will, on each layer, provide a cluster decomposition where in each unfrozen cluster, there are no cuts of cut-size at most $\lmax$ which are $(1-\epsilon)$-boundary-sparse.

We will then maintain this property on each layer, by realizing that each intracluster edge insertion and intercluster edge deletion cannot create any $(1-\epsilon)$-boundary-sparse cut, while for any other update, only a limited number of calls to LocalKCut suffice to find all newly created $(1-\epsilon)$-boundary-sparse cuts.

In fact, we will simply mark a few nodes after each update as unchecked, so as to make sure that \Cref{inv:cuttingisgood} holds, and then run the corresponding part of the algorithm again.
The goal is to have after each update a cluster decomposition as in \Cref{lem:clusterdecomposition} if the current minimum cut value ${\Tilde \lambda}$ satisfies $\lmin \le {\Tilde \lambda} \le \lmax$.

Remember that we set $\epsilon = \frac 1 {\sqrt{\log n}}$, $\alpha = \frac{1}{\poly \log n}$, $\phi = 2^{-\Theta(\log^{3/4} n)}$, and $\gamma = \frac 1 3$.
Note that $\frac \alpha \phi \ge 1$ and that we assume $\eps \le 0.04$.
We further set $\psi=2^{\Theta(\log^{1/2} n)}, h= \Theta(\log^{1/2}m),$ and $\rho = 2^{\Theta(\log^{1/2} n)}$.

We describe below the main structure of the dynamic algorithm:

\begin{algo}\label{alg:dynamicestimate}
Data Structure: The data structure is the same the one of the static algorithm (\Cref{alg:StaticAlgorithm}). 
In particular, we have a cluster decomposition, the corresponding mirror graphs, a Buffer instance (\Cref{alg:buffer}) for each of those mirror graphs, and we maintain the mirror cuts (\Cref{alg:maintainmirrors}) on each mirror graph.

\begin{enumerate}[noitemsep]
  \item If the graph has 2 nodes, output the number of edges between the two nodes. If the graph has one node, output ``$\tilde \lambda > \lmax$''. In any case, stop this call to the algorithm.
    \item  We run \Cref{alg:dynclusterdecompose}.
    \Cref{alg:dynclusterdecompose} is the dynamic algorithm that maintains the cluster decomposition such that each unfrozen cluster $C$ contains no cut $S\subsetneq C$ that is $(1-\epsilon)$-boundary-sparse and satisfies $w(S, C\setminus S) \le \lmax$ and $\vol (S) \le \frac \lmax \phi$, and the corresponding mirror clusters. 
    \item  We then maintain the contracted graph where each cluster output by \Cref{alg:dynclusterdecompose} is contracted into a node, and call recursively \Cref{alg:dynamicestimate} on that contracted graph.
    \item We then output the value of the minimum cut found among the cut returned by the recursive call and the mirror clusters, if that value falls in $[\lmin, \lmax]$.  We store pointers to the corresponding cut. If a cluster was frozen or we found a cut of value at most $\lmin$, we report that the mincut is strictly smaller than $\lmin$, if the smallest cut found is larger than $\lmax$ (or no cut is found), we report that the mincut is strictly larger than $\lmax$.
\end{enumerate}
   
\end{algo}

Since the cluster decomposition satisfies the conditions of \Cref{lem:clusterdecomposition} at every point, this proves correctness.

The subsections below will prove the following lemma:

\begin{restatable}{lemma}{lemrecourse}\label{lem:clusterdec}
In \Cref{alg:dynclusterdecompose}:
\begin{enumerate}[noitemsep]
    \item  The amortized recourse\footnote{the recourse is the number ofcmodifications to the edges of the contracted graph} is $\Tilde O\PAR{\frac \rho \epsilon 2 ^{\frac { O(1) } \epsilon}}$
    \item At every time step, and every level, if $\tilde m$ is the number of edges input on that level, the number of intercluster edges is at most $\Tilde{O} \PAR{\tilde m\phi \frac \rho \epsilon 2 ^{\frac { O(1) } \epsilon}}$, and thus the recursive call will have as input at most $\Tilde{O} \PAR{\tilde m\phi \frac \rho \epsilon 2 ^{\frac { O(1) } \epsilon}}$ edges.
    \item The amortized update time is $n^{o(1)}$ on each level.
    \item At every given time, the algorithm is correct with probability at least $1-n^{-4}$.
\end{enumerate}
\end{restatable}

\begin{lemma}\label{lem:nocalls}
    There are  $r=O(\log^{\frac 1 4} n) $ recursive calls in \Cref{alg:dynclusterdecompose}.
\end{lemma}

\begin{proof}
    If a graph $G$ has $m$ edges, by \Cref{lem:clusterdec}, the recursive call on its contracted graph will have $\Tilde{O} \PAR{m\phi \frac \rho \epsilon 2 ^{\frac { O(1) } \epsilon}}$ many edges.
    Hence, the number of input edges gets multiplied by at most $\Tilde{O} \PAR{\phi \frac \rho \epsilon 2 ^{\frac { O(1) } \epsilon}}$ at each recursive call.
    And thus, after $r$ many recursive calls, we are down to $m \PAR{\phi \frac \rho \epsilon 2 ^{\frac { O(1) } \epsilon}}^r$ edges.
    For $r=O(\log^{\frac 1 4} n)$ with large enough constant factor, we get that $m \PAR{\phi \frac \rho \epsilon 2 ^{\frac { O(1) } \epsilon}}^r \le 1$.
\end{proof}
It follows that there are $O(\log^{\frac 1 4} n)$ levels of recursion.
\begin{lemma}
    The amortized recourse computed over all recursion levels is $n^{o(1)}$.
\end{lemma}

\begin{proof}
    The total recourse on the $i$-th level is $\Tilde O\PAR{(\frac \rho \epsilon 2 ^{\frac { O(1) } \epsilon})^{i-1}}$, as per \Cref{lem:clusterdec}.
    As $i= O(\log^{\frac 1 4} n)$ by the lemma above, we get that $\Tilde O\PAR{(\frac \rho \epsilon 2 ^{\frac { O(1) } \epsilon})^{i-1}} = 2^{O(\log^{\frac 3 4}n)}=n^{o(1)}$.
    Summing over all levels of recursion gives us the desired result.
\end{proof}

\begin{lemma}
    The amortized running time over all levels is $n^{o(1)}$ in \Cref{alg:dynclusterdecompose}.
\end{lemma}

\begin{proof}
    Each recursion level receives an amortized number of $n^{o(1)}$ updates as per the previous lemma, and requires $n^{o(1)}$ amortized time to process the update, by \Cref{lem:clusterdec}.
    Hence, on each level, after one update on the top level, the time to process that update is $n^{o(1)}$.
    Summing over all levels gives us the desired result.
\end{proof}

\begin{lemma}
    In \Cref{alg:dynclusterdecompose}, the total approximation ratio after all calls is $1+O(\frac 1 {\log ^{\frac 1 4}n})$.
\end{lemma}

\begin{proof}
    Each recursive call degrades the quality of the solution by a factor $(1+2 \epsilon)$ at most, by \Cref{lem:clusterdecomposition}. The total approximation ratio is 
    \begin{align*}
        (1+2 \epsilon)^r &= (1+\frac {O(1)} {\sqrt{\log n}})^ {O(\log^{\frac 1 4} n)} \\&= \exp\PAR{(\log(1+\frac {O(1)} {\sqrt{\log n}})O(\log^{\frac 1 4} n)} \le \exp\PAR{\frac{O(\log^{\frac 1 4} n)}{\sqrt{\log n}}} = 1+O(\frac 1 {\log^{\frac 1 4}n})
    \end{align*}
\end{proof}

\begin{lemma}
    At every given time, \Cref{alg:dynamicestimate} is correct with probability at least $1-n^{-3}$.
\end{lemma}

\begin{proof}
    This is a union bound over all levels, of which there are $n^{o(1)}$ by \Cref{lem:nocalls}.
\end{proof}

\subsection{Proof of \Cref{lem:clusterdec}}
In this subsection, we are given a graph with updates, and our goal is to maintain a cluster decomposition of this graph such that no cluster contains a cut of cut-size at most $\lmax$ that is $(1-\epsilon)$-boundary sparse.
As in the static setting the main idea is to maintain a dynamic expander decomposition which we refine further. 

\begin{definition}
    We say that an edge $e=(u,v)$ is \emph{incident} to a cluster $C$ if either $u$ or $v$ (or both) are in $C$.
\end{definition}

\paragraph{Note on frozen clusters.} Note that in this algorithm, a cluster can become frozen at two different times: either during the preprocessing, when the static algorithm is run, or when handling updates, that is after running \Cref{alg:subroutine}.

\begin{algo}[Maintaining the cluster decomposition]\label{alg:dynclusterdecompose}
    Preprocessing: We run the static algorithm on the initial graph, with the preprocessing of \Cref{thm:expanderdecomposition} for the expander decomposition.

    Handling updates: 
    We feed the next $O(\frac{m\phi}{\rho})$ updates to the dynamic expander decomposition algorithm.
    Every change output by that algorithm is then processed in the following way:
    \begin{enumerate}[noitemsep]
        \item Mark the endpoints of the updated edge as unchecked.
        \item If an edge turns from being an intraexpander edge to an interexpander edge\footnote{These edges and the updated edge are called affected edges.}, we update the data structure as discussed in \Cref{sec:datastructure}. Note that this includes updates to instances of \Cref{alg:buffer} to detect if any cluster should change from frozen to unfrozen, and to the instances of the Mirror Cuts Algorithm (\Cref{alg:maintainmirrors}) to maintain the mirror cuts. 
        \item  Then, for each affected unfrozen cluster $C$ (i.e. containing an endpoint of an affected edge), find $2\lmax$ intercluster edges leaving that cluster whose endpoint in $C$ is checked (or all of them if there are fewer than $2\lmax$ such edges), and mark their incident vertices in the cluster as unchecked.
        \item While there exists in $\C$ (the cluster decomposition) an unfrozen cluster $C$ with an unchecked vertex, we run the Find and Cut subroutine (\Cref{alg:subroutine}) on $C$, which might change $\C$.
    Remember that if it splits a cluster, it updates the data structure as discussed in \Cref{sec:datastructure}, as well as the corresponding instances of the Buffer Algorithm and the Mirror Cuts Algorithm (\Cref{alg:buffer} and \Cref{alg:maintainmirrors}) .
    \end{enumerate}

    Every $O(\frac{m\phi}{\rho})$ updates, we restart from scratch, i.e., we run the preprocessing step on the full graph.
\end{algo}

We start by showing the correctness of the algorithm.

\subsubsection{Correctness}

\begin{lemma}\label{lem:keepinvariant}
    Let $C$ be a cluster that statisfies \Cref{inv:cuttingisgood} at time $t$.
    Assume that between times $t$ and $t'$, $C$ is subject to updates, and that the Find and Cut subroutine (\Cref{alg:subroutine}) is not run in between $t$ and $t'$ on the cluster\footnote{This can either be because the cluster is frozen, or because $t'$ is the first time after $t$ where an update is incident to $C$}.
    If we mark every node incident to an edge insertion or deletion as well as the nodes incident to $2\lmax$ intercluster edges incident to $C$ apart from the edges inserted between $t$ and $t'$ (or simply all of the intercluster edges incident to $C$ if there are fewer than $2\lmax$ of them)  as unchecked, then $C$ satisfies \Cref{inv:cuttingisgood} at time $t'$.
\end{lemma}
\begin{proof}
    Let $S$ be, at time $t'$, a $(1-\epsilon)$-boundary-sparse cut in $C$, with $w(S, C\setminus S) \le \lmax$, and $\vol (S) \le \frac \lmax \phi$.
    We need to show that $S$ contains an unchecked vertex at time $t'$. 
    We have multiple cases:
    
    Case 1: $S$ contains an unchecked vertex at time $t$:

    In that case, the same vertex is still unchecked at time $t'$.
    Indeed, since the Find and Cut subroutine was not run between $t$ and $t'$, no unchecked vertex becomes checked between the two times.

    Case 2: There is an edge with at least one endpoint in $S$ that was inserted or deleted.
    That update causes the corresponding endpoint to become unchecked. 
    Since the Find and Cut subroutine was not run between $t$ and $t'$, this endpoint is still unchecked at time $t'$.

    Case 3: If all vertices of $S$ were checked at time $t$ and no edge update was incident to $S$ between $t$ and $t'$:
    It follows by \Cref{inv:cuttingisgood} that $S$ was either:
    \begin{enumerate}[label=(\alph*), noitemsep]
        \item Not $(1-\epsilon)$-boundary sparse at time $t$, or
        \item $w(S, C\setminus S) > \lmax$ at time $t$, or
        \item $\vol (S) > \frac \lmax \phi$ at time $t$.
    \end{enumerate}
    As at time $t'$ we have that $w(S, C\setminus S)$ and no edge update was incident to $S$ between $t$ and $t'$, case (b) is impossible.

    As at time $t'$ we have that $\vol (S) \le \frac \lmax \phi$ and no edge update was incident to $S$ between $t$ and $t'$, case (c) is impossible.

    Let us thus look at case (a): At time $t$, $w(S, C\setminus S) \ge (1-\epsilon)\cdot \min\{w(S,V\setminus C), w(C\setminus S, V\setminus C)\}$.

    Since no edge update was incident to $S$ between $t$ and $t'$, we have that $w(S, V\setminus C)$ and $w(S, C\setminus S)$ remain unchanged between $t$ and $t'$.

    At time $t'$, by boundary sparseness, we have that $w(S, C\setminus S) < (1-\epsilon) w(S, V\setminus C)$.
    This inequality also holds at time $t$, and since $S$ is not $(1-\epsilon)$-boundary sparse at time $t$, we must have that at time $t$: $w(S,C\setminus S) \ge (1-\epsilon) w(C\setminus S, V\setminus C)$.

    But, as we are not in case (b), at time $t$, we have that $\lmax \ge w(S, C\setminus S) \ge (1-\epsilon) w(C\setminus S, V\setminus C)$.
    Therefore $w(C\setminus S, V\setminus C) < \frac 1 {1-\epsilon} \lmax<2\lmax$ for $\epsilon \in [0, \frac 1 2]$. 

    We have two cases: (i): $E(C, V\setminus C)$ contains at least $2\lmax$ many edges.
    As at least $2\lmax$ many edges in $E(C, V\setminus C)$ have their endpoints in $C$ unchecked, and there are strictly less than $2\lmax$ in $E(C\setminus S, V\setminus C)$, at least one edge in $E(S, V\setminus C)$ has its endpoint in $C$ unchecked. 
    That endpoint is in $S$ which concludes the proof.

    (ii): $E(C, V\setminus C)$ contains at most $2\lmax$ many edges.
    Since all of those edges have their endpoints unchecked, it remains to show that one of them has an endpoint in $S$.
    This is straightforward as $S$ is boundary sparse at time $t'$, which implies that $w(S, V\setminus C) > \frac 1 {1-\epsilon} w(S, C\setminus S) \ge 0$.    
\end{proof}

\begin{corollary}\label{cor:renewinvariant}
    Assume \Cref{inv:cuttingisgood} holds before some point in time $t$, and assume we have $t'-t$ updates.
    By marking at time $t'$ as unchecked every node adjacent to an edge insertion or deletion, as well as the nodes adjacent to $2\lmax$ intercluster edges apart from the edges inserted to a frozen cluster (or all of the intercluster edges if there are not $2\lmax$ of them), the cluster satisfies \Cref{inv:cuttingisgood} when it is unfrozen.
\end{corollary}

\begin{remark}
    One can apply \Cref{lem:keepinvariant} for exactly one update between $t$ and $t'$, which is what we do when the cluster is unfrozen both before and after the update.
\end{remark}

\begin{corollary}\label{dynproba}
Assume that all instances of \Cref{alg:buffer} are correct, and that \Cref{inv:cuttingisgood} holds at some time $t$.
Assume the graph is subject to $t'-t \le n^2$ updates.
    If there are no cuts in $G$ of cut-size strictly smaller than $\lmin$ at time $t'$, then  with probability at least $1-n^{-6}$, in the cluster decomposition maintained by \Cref{alg:dynclusterdecompose} at time $t'$ no cluster $C$ contains a cut $S$ that is $(1-\epsilon)$-boundary-sparse and such that $w(S, C\setminus S) \le \lmax$ and $\vol (S) \le \frac \lmax \phi$.
\end{corollary}
    
\begin{proof}
This is a direct consequence of \Cref{lem:3.13} and \Cref{cor:renewinvariant}, as since there are no cuts at time $t'$ in $G$ that are strictly smaller than $\lmin$, we are ensured that no cluster is frozen, and \Cref{alg:dynclusterdecompose} ensures no unchecked vertex remains in any unfrozen cluster. We only have to consider the probability statemen.

    After each update, let us look at each vertex $v$.
    We can run a batch of LocalKCut at most once per while loop of \Cref{alg:dynclusterdecompose}.
    For a vertex to be in another while loop, the cluster it is in must have been split.
    Since the cluster can be split at most $n$ times, each vertex is subject to at most $n$ batches of LocalKCut, and thus a batch of LocalKCut is run at most $n^2$ times per update, and thus $n^4$ overall. 
 Since each batch finds the necessary cut with probability at least $n^{-10}$ by \Cref{lem:chernoff}, a union bound completes the proof.
 \end{proof}

This, together with \Cref{lem:clusterdecomposition} shows correctness. It remains to compute the probability it fails between two rebuilds.

\begin{lemma}\label{lem:proba}
    At any given point in time, \Cref{alg:dynclusterdecompose} is correct with probability at least $1-n^{-4}$.
\end{lemma}

\begin{proof}
    Let us analyze the probability of failing of each of the components of \Cref{alg:dynclusterdecompose}, between two rebuilds.

    First, the preprocessing is correct with probability at least $1-3n^{-5}$, by \Cref{cor:3.15}.
    Then note that a rebuild happens every $O(\frac {m\phi} \rho) = O(n^{2 - o(1)})$ updates.

    During these many updates, a volume of at most $n^3$ updates can be applied to an instance of  the Buffer Algorithm (\Cref{alg:buffer}) or the Mirror Cuts Algorithm \Cref{alg:maintainmirrors}.
    Indeed, apart from the $n^2$ many updates to the graph that can happen to be in the cluster itself, it is possible that the cluster is subject to splits.
    However a volume of at most $n^2$ can be split off from the cluster, leaving us with $2n^2$ updates at most.

    Hence, by \Cref{cor:bufferoutput} and \Cref{cor:mirrorcutsoutput}, both algorithms are correct with probability at least $1-n^{-6}$.
    A union bound over the at most $n$ clusters show that they are all correct with probability at least $1-2n^{-5}$.

    The final element that is randomized is the calls to LocalKCut in the while loop of \Cref{alg:dynclusterdecompose}. 
    As per \Cref{dynproba}, once we know that \Cref{alg:buffer} is correct, all the calls to LocalKCut are correct with probability at least $1-n^{-6}$.
    
    A union-bound over those probabilities give the desired result, as $2n^{-5} +n^{-6} \le n^{-4}$.
\end{proof}

\subsubsection{Bounding the number of intercluster edges}

We now take a look at the number of intercluster edges just before a new rebuild, and aim to bound that number.
For that, as for the static algorithm we build below a forest of all the clusters that existed since the last rebuild, and analyze carefully the number of intercluster edges at each node of this forest.

Formally, let $M$ be the initial number of interexpander edges (after running the first static expander decomposition) and $R$ be the number of changes output by the dynamic expander decomposition.
We aim to bound the number of intercluster edges after all of the updates before the next rebuild.

For that, we build a forest $F$ of all the clusters ever created.
The parent of each cluster is the cluster it was cut off from when it was created.
The roots are the expanders as output by the static expander decomposition.
We will denote by $t(C)$ the \emph{time} of cluster $C$, that is, the update at which $C$ was split (if it is an internal node), or the last update before the rebuild otherwise, and $\boundary C(t)$ the boundary of cluster $C$ after update $t$.
Let $w_t(A,B)$ be the number of edges between $A$ and $B$ at time $t$.
Here, we assume that each update corresponds to one change output by the dynamic expander decomposition.

\begin{remark}
    For the rest of this section, we will assume that we always split a cluster along a $(1-\epsilon)$-boundary sparse cut.
    This is obviously not necessarily the case, as a new split might appear either because we found a $(1-\epsilon)$-boundary sparse cut, or because of an update in the expander decomposition, which results in new interexpander edges.
    However, we can simulate that second case with the first one: as the expander decomposition outputs all the edges that need to be considered interexpander now, we can start by deleting all these edges from our graph, then inserting two edges, one on each side of the new cluster to be split. 
    That is: If $C$ is a cluster in $G$, and the expander decompsition requires $C$ to be split into $C$ and $C\setminus S$, we delete all of the edges between $C$ and $C\setminus S$, then add two edges, one from $S$ to $V\setminus C$, and one from $C\setminus S$ to $V\setminus C$.
    This ensures that $S$ is $(1-\epsilon)$-boundary sparse, and we can now cut along it.
    We then insert all of the edges between $C\setminus S$ and $S$ back, and delete the two added edges.
\end{remark}

\begin{lemma}\label{lem:dynboundariessmaller}
    Let $C$ be a node in $F$ and $S, C\setminus S$ its children.
    Then we have that $\boundary C(t(C)) \ge \max\{\boundary S (t(C)), \boundary (C\setminus S)(t(C))\}+\frac{\epsilon \lmin} 2$.
\end{lemma}

\begin{proof}
    In this proof, all quantities are considered at time $t(C)$.
    
    Since $S$ is a $(1-\epsilon)$-boundary-sparse cut, we have that $w(S, C\setminus S) \le (1-\epsilon)\min\{ w(S, V \setminus C), w(C\setminus S, V\setminus C)\}$.
    Therefore, $\boundary S = w(S, V\setminus C) + w(S, C\setminus S) \le w(S, V\setminus C) + (1-\epsilon) w( C\setminus S, V\setminus C) \le \boundary C - \epsilon w( C\setminus S, V\setminus C)$.

    But $w( C\setminus S, V\setminus C) + w( C\setminus S, S) = \boundary (C\setminus S) \ge \lambda_{\min}$.
    Since $w( C\setminus S, V\setminus C) > w( C\setminus S, S)$ (by boundary-sparsity of $S$), we have that $w( C\setminus S, V\setminus C) \ge \frac {\lambda_{\min}} 2 $ and the result follows by symmetry for $\boundary (C\setminus S)$.    
\end{proof}

For each cluster $C$ let $L(C)$ be its boundary size when it was split into two (if it is an internal cluster) or by its boundary size after all updates before the rebuild (if it is a leaf).

We will then, for the analysis, add marks to each cluster as follows: A cluster $C$ that is a root gets mark $\oplus$ if $L(C)\ge 3\lmax$ and $\odot$ otherwise.
Then, in BFS fashion, each cluster $C$ gets marked:
\begin{itemize}[noitemsep]
    \item if its parent has $\oplus$: $\oplus$, if $L(C) \ge 3\lmax$, and $\odot$ otherwise,
    \item if its parent has $\odot$ or $\ominus$: $\oplus$ if $L(C) \ge 6\lmax$, $\ominus$ otherwise.
\end{itemize}

We first analyze the forest $F'$ which consists of the union of the maximal subtrees whose roots and internal nodes are marked $\oplus$, and leaves are marked $\odot$ or $\oplus$.
The goal is to show that the sum of the labels of the leaves is at most $O(\frac {M+R} \epsilon)$.
As $M= O(m\phi)$ and $R=O(m\phi)$, this will yield at most $O(\frac {m\phi} \epsilon)$ many intercluster edges, which is crucial for bounding the recourse, as well as the running time (which depends on the number of intercluster edges).

We look at the following potential:
$$
\Phi(C) = \max\{0, \boundary C(t(C)) - 2.1 \lmax\}
$$
and the total potential:
$$
\Phi\{\C\} = \sum_{c\in\C} \Phi(C)
$$

We start with $\C=\C_0$ being the roots of $F'$ and at each step, we will replace one element from $\C$ with its two children until $\C=\C_f$ only contains leaves of $F'$.

\begin{lemma}\label{lem:lemtwo}
    Let $C \in F'$ be a set in $\C$ replaced with $S$ and $C\setminus S$.
    Let $u(S)$ and $u(C\setminus S)$ be the number of updates to the boundary of $S$ and $C\setminus S$ respectively for updates between $t(C)$ (excluded) and $t(S), t(C\setminus S)$ (included).

    Then $\Phi(S) + \Phi(C\setminus S) \le \Phi(C) - \frac{\epsilon \lmin}{2} + u(S)+u(C\setminus S)$.
\end{lemma}

\begin{proof}
    Recall that $\boundary C(t(C)) \ge 3\lmax$, which imply that $\Phi(C) \ge 0.9\lmax$.
    We then have three cases:
    \begin{itemize}
        \item if both $\boundary S(t(S)), \boundary (C\setminus S)(t(C\setminus S))\le 2.1\lmax$:
        $$
        \Phi(S) + \Phi(C\setminus S) = 0 \le \Phi(C)  - \frac {\epsilon \lmax} {2}
        $$
        as $\Phi(C) - \frac {\epsilon \lmax}{2} \ge 0.9\lmax - \frac {\epsilon \lmax}{2} >0$.
        \item if both $\boundary S(t(S)), \boundary (C\setminus S)(t(C\setminus S))> 2.1\lmax$:
        \begin{multline*}
        \Phi(S) + \Phi(C\setminus S) = \boundary S(t(S)) +\boundary (C\setminus S)(t(C\setminus S)) - 2\times 2.1\lmax\\
        \overset{(i)}\le \boundary S(t(C))+\boundary (C\setminus S)(t(C)) -2\times 2.1\lmax + u(S)+u(C\setminus S)\\
        \le \boundary C(t(C)) +2w_{t(C)}(S, C\setminus S) -2\times 2.1\lmax + u(S) + u(C\setminus S)\\
        \le \Phi(C)  +2\lmax -2.1\lmax +u(S)+u(C\setminus S)
        \le \Phi(C) - \frac {\epsilon \lmax}{2} +u(S)
        u(C\setminus S)
        \end{multline*}
        Where (i) stems from the fact that between times $t(C)$ and $t(S)$, at most $u(S)$ many edges got added to $S$ and therefore $\boundary S$.
        \item if wlog $\boundary S(t(S))\le 2.1\lmax$ and $\boundary (C\setminus S)(t(C\setminus S))> 2.1\lmax$:

    \begin{align*}
        \Phi(S)+\Phi(C\setminus S)&= 0 + \boundary(C\setminus S)(t(C\setminus S))-2.1\lmax\\
        &\le 0 + \boundary(C\setminus S)(t(C))+u(C\setminus S)-2.1\lmax\\
        &\le \boundary C(t(C))-w_{t(C)}(S, V\setminus C) + w_{t(C)}(S, C\setminus S)+u(C\setminus S)-2.1\lmax\\
        &\overset{(ii)}\le \Phi(C)-\epsilon w_{t(C)}(S, V\setminus C) +u(C\setminus S)\\
    \end{align*}
    Where $(ii)$ follows from the $(1-\epsilon)$ boundary-sparsity of $S$ at time $t(C)$.
    
    We conclude using the following claim:
    \begin{claim}
        For any cluster $C$ and $(1-\epsilon)$-boundary sparse cut $S\subsetneq C$, we have that either $w(S, V\setminus C) \ge \frac \lmin 2$ or $\boundary S < \lmin$.
    \end{claim}
    \begin{proof}
        Assume we have that $w(S, V\setminus C) < \frac \lmin 2$.
        We are going to show that in that case,  $\boundary S < \lmin$ holds.
        By boundary sparsity, we have that $w(S, C\setminus S)<(1-\epsilon) w(S, V\setminus C)<(1-\epsilon)\frac \lmin 2 $, and therefore, $\boundary S = w(S, C\setminus S) + w(S, V\setminus C)<\lmin$.       
    \end{proof}
    In our case, since we do not cut along $S$ if $\boundary S<\lmin$, we have that $\boundary S \ge \lmin$ and thus $w(S, V\setminus C) \ge \frac \lmin 2$.
    \end{itemize}
\end{proof}

\begin{lemma}\label{lem:sumfprime}
    The sum of all the labels of the leaves of $F'$ does not exceed $O\PAR{\frac {M+R} {\epsilon}}$.
\end{lemma}

\begin{proof}
    Let us consider one cluster $C_0$ that is a root in $F'$, and we consider only its component in $F'$. We start with $\Phi(C_0)=O(\boundary C_0(t_0)+u(C_0))$, where $t_0$ is the time of the last rebuild.
    By \Cref{lem:lemtwo}, the potential decreases by at least $\frac {\epsilon\lmin} 2 -u(S) - u(C\setminus S)$ every time we replace a cluster $C$ by its children $S$ and $C\setminus S$.
    Therefore, after $X$ many replacements each of which increases the number of clusters by 1, the potential is at most $\Phi(C_0) -X\frac {\epsilon \lmin} 2 + \sum_{C \in F'(C_0)} u(C) \le \Phi(C_0)  +U(C_0)-X\frac {\epsilon \lmin} 2$, where $F'(C_0)$ is the component of $C_0$ in $F'$ and $U(C_0)=\sum_{C \in F'(C_0)} u(C)$.
    The potential being positive, we have that $X=O\PAR{\frac{\Phi(C_0) + U(C_0)} {\epsilon \lmin}}$, and thus we have $O\PAR{\frac{\Phi(C_0) + U(C_0)} {\epsilon \lmin}}$ many clusters descendents of $C_0$ in $F'$.

    Let $\C_f(C_0)$ be the set of all leaves in the component of $C_0$ in $F'$.
    Recall that $L(C) = \boundary C(t(C)$.
    We then have:
    {\renewcommand{\boundary}[1]{L(#1)}
    \begin{align*}
        \sum_{C \in \C_f(C_0)} \boundary C &= \sum_{\substack{C \in \C_f(C_0)\\\boundary C \le 2.2\lmax}} \boundary C+\sum_{\substack{C \in \C_f(C_0)\\\boundary C > 2.2\lmax}} \boundary C\\
        \intertext{In the first sum, we know the that the total number of clusters is at most $O\PAR{\frac{\Phi(\C_0) + U(C_0)} {\epsilon \lmin}}$. For the second sum, note that if $\boundary C \ge 2.2\lmax$, then $\boundary C = O(\Phi(C))$.}
        &\le  O\PAR{\frac{\Phi(\C_0) + U(C_0)} {\epsilon \lmin}} 2.2\lmax+\sum_{\substack{C \in \C_f(C_0)\\\boundary C > 2.2\lmax}} O(\Phi( C))\\
        &\le  O\PAR{\frac{\Phi(\C_0) + U(C_0)} {\epsilon }}+ O(\Phi(\C_f(C_0))) = O\PAR{\frac{\Phi(\C_0) + U(C_0)} {\epsilon }}\\
    \end{align*}}

    For the last equality, recall that the sum of potentials of the two children of a node in $F'$ is less than the potential of the parent to which we add the number of updates made to the children. It follows by induction that for any set of descendants of $C_0$ such that none is an ancestor of the other, the sum of their potentials is less than the potential of $C_0+U(C_0)$.
    As the clusters in $C_f(C_0)$ fulfill this condition, it follows that
    $\sum_{\substack{C \in \C_f(C_0)\\\boundary C > 2.2\lmax}} O(\Phi( C)) = O(\Phi(\{C_0\}+U(C_0))$.

    Summing over all roots of $F'$, we get that the sum of the boundaries of all leaves of $F'$ does not exceed $O(\Phi(\{\C_0\}+U(C_0))$, where $\C_0$ is the set of all roots of $F'$.

    Let us now estimate $\Phi(\C_0)$.
    In $\C_0$, we first have the roots in $F'$ that are an output from the static expander decomposition.
    The sum of the boundaries of these expanders is $M$ initially, and does not exceed $M+R$ when they are cut.

    Then there are also the clusters $C$ marked $\oplus$ whose parent $P$ is marked either $\odot$ or $\ominus$, as those are also roots in $F'$.
    Since $P$ is marked $\odot$ or $\ominus$, we know that $L(P) \le 3\lmax$, while $L(C) \ge 6\lmax$.
    By \Cref{lem:lemtwo}, this implies that $u(C) \ge 3\lmax$, and therefore $\boundary C (t(C)) \le \boundary C(t(P)) + u(C) \le 3\lmax + u(C) = \Theta(u(C))$.
    Hence, $\Phi(C) \le \boundary C(t(C)) = \Theta(u(C))$.
    Summing over all such clusters, we get $\sum_{\substack{C\in \C_0\\C \mathrm{ not expander}}} \Phi(C) = \Theta (R)$.

    And therefore $\Phi(\C_0) = O(M+R)$.    
    \end{proof}

    We now analyze the forest $F''$ that is the union of all maximal subtrees whose root is marked $\odot$ and whose nodes are marked $\ominus$.

    We know, from \Cref{lem:sumfprime}, that the total boundary size of its roots is at most $O\PAR{\frac {M+R} {\epsilon}}$, as every root of $F''$ is a leaf of $F'$.
    Moreover, no boundary individually exceeds $6\lmax$.

    \begin{lemma}\label{lem:boundaries}
        The sum of all the labels of the leaves of $F''$ does not exceed $O\PAR{\frac {M+R} {\epsilon} 2^{\frac {O(1)} \epsilon}}$.
    \end{lemma}

    \begin{proof}
        By \Cref{lem:dynboundariessmaller}, we know that, for any internal node $C$ with children $S$ and $C\setminus S$, we have $L(S), \le L(C) - \frac{\epsilon \lmin} 2 + u(S)$ and $L(C\setminus S) \le L(C) - \frac{\epsilon \lmin} 2+u(C\setminus S)$.
        
        Further mark by $\bullet$ nodes $S$ that satisfy $L(S) \le L(C) - \frac {\epsilon \lmin} 4$ and by $\circ$ the others.

        We have at most $\frac {R}{\frac {\epsilon \lmin} 4} = O(\frac R {\epsilon \lmin})$ many nodes labeled $\circ$.
        Look at any such node, and all of its $\bullet$ descendants that are reached by a path of only $\bullet$ nodes. This forms a tree.
        The tree has at most $2^{\frac{6\lmax}{\frac {\epsilon \lmin} 4}} \le 2^{\frac {O(1)} \epsilon}$ nodes of boundary at most $6\lmax$ each.
        Overall this yields a total boundary size of $O\PAR{\lmax 2^{\frac{O(1)} \epsilon}}$ for the tree, and summing over all $O(\frac R {\epsilon \lmin})$ trees we get $O\PAR{\frac R \epsilon 2^{\frac{O(1)} \epsilon}}$ maximum total boundary size.

        We should also account for all of the nodes marked $\bullet$ who do not have a $\circ$ ancestor. Let $C$ be the root of a subtree of $\bullet$ nodes who have no $\circ$ ancestor in $F''$. 
        The number of nodes in that subtree is at most $2^{\frac{\boundary C}{\frac {\epsilon \lmin} 4}} \le 2^{\frac{O(1)} \epsilon} $.
        The total boundary size of all nodes in that tree is then $\boundary C 2^{\frac{O(1)} \epsilon}$.
    Summing over all those subtrees we get that the total boundary size is at most $\sum_{\substack{C:\bullet\\C \text{ root in } F''}} \boundary C 2^{\frac {O(1)} \epsilon} = O(\frac {M+R} \epsilon 2^{\frac {O(1)} \epsilon})$.

    Summing the two values, we get the result.
    \end{proof}

\subsubsection{Running time}

\begin{lemma}
    The total number of times, over the $O(\frac{m\phi}\rho)$ updates to the input of \Cref{alg:dynclusterdecompose}, we process an unchecked vertex is $O(\frac {M+R} \epsilon 2^{\frac {O(1)} \epsilon} \frac \lmax \phi)$.
\end{lemma}

\begin{proof}
    Let us first estimate the number of nodes that get marked as unchecked overall (counting a node that has been unchecked multiple times with multiplicity):

    We divide the nodes that were unchecked at some point in two categories: nodes that were adjacent to an edge update, and nodes that we adjacent to an edge that became an intercluster edge.
    Any edge that was an intercluster edge and then deleted belongs to the first category.

    Nodes that were adjacent to an edge update are at most $R$ many.
    Nodes that were adjacent to an edge that was cut are at most $O(\frac {M+R} \epsilon 2^{\frac {O(1)} \epsilon})$ many, as there are $O(\frac {M+R} \epsilon 2^{\frac {O(1)} \epsilon})$ many such edges at the final decomposition, and any intercluster edge in a prior decomposition is either a deleted edge or an intercluster edge in the final decomposition, as we do not merge clusters.

    In the worst case, all updates are made to an unfrozen cluster in $G$, and thus we need to add $2\lmax$ unchecked nodes.
    Moreover, while processing an unchecked node, we might find a cut that needs to be split from the original cluster.
    In that case, we have to process the node again, but this can happen at most $\frac 1 \phi$ times as per \Cref{lem:nottoomanykargers}.
\end{proof}

\begin{corollary}\label{cor:time}
    As $\lmax = n^{o(1)}$ and $\phi = n^{-o(1)}$, the total running time of \Cref{alg:dynclusterdecompose} between two rebuilds is $(M+R)\cdot n^{o(1)}$.
\end{corollary}

\begin{proof}
    As per \Cref{lem:kargertime}, the running time of each call to LocalKCut is $\Tilde O(\frac \lmax \phi)$, and therefore all calls to LocalKCut need $\Tilde O(\frac {M+R} \epsilon 2^{\frac {O(1)} \epsilon} \frac {\lmax^{12}} {\phi^8})$ time. 
    This is also an upper bound to the number of nodes that need to be unchecked, as each LocalKCut is ran on an unchecked vertex.
    Furthermore, we need to update the data structure, which takes amortized time $n^{o(1)}$ per update, by \Cref{lem:structuretime}.
    Moreover, the calls to the Buffer Algorithm and the Mirror Cuts Algorithm (\Cref{alg:buffer} and \Cref{alg:maintainmirrors})  also time $m^{1+o(1)}$ between two rebuilds, by \Cref{cor:bufferoutput} and \Cref{cor:mirrorcutsoutput}, as the graph has initial volume $m$ and is subject to updates of volume $m\phi$ between two rebuilds.  time over all the updates between two rebuilds yields the desired result. 
\end{proof}

This proves the following lemma:

\lemrecourse*

\begin{proof}
    Item 1. is proven using \Cref{lem:boundaries} where we use $M=O(m\phi)$ and $R=O(m\phi)$ as output by the expander decomposition (see \Cref{thm:expanderdecomposition}). In \Cref{lem:boundaries}, we show that between two rebuilds, which happen every $\Theta(\frac {m \phi} {\rho})$ with $\phi=2^{-\Theta (\log^{3/4}n)}$ and $\rho=2^{-\Theta (\log^{1/2}n)}$, we have at most $O\PAR{\frac {m\phi} {\epsilon} 2^{\frac {O(1)} \epsilon}}$ intercluster edges.
We can issue an insertion and a deletion for each of those edges, and amortize them over all updates between two rebuilds.
This yields a recourse of $O\PAR{\frac \rho \epsilon 2^{\frac {O(1)} \epsilon}} = 2^{O(\log^{1/2}n)}$ as $\epsilon = \frac 1 {\sqrt {\log n}}$.

    Item 2. is a direct applicaation of  \Cref{lem:boundaries}, where we note that between two rebuilds, the decomposition is only refined, and thus the number of intercluster edges just before the next rebuild is an upper bound on the number of intercluster edges at any point since the last rebuild.
    
    Item 3. is proven in \Cref{cor:time}.
    
    Item 4. is proven in \Cref{lem:proba}.
\end{proof}

\subsection{Our Data Structure}\label{sec:datastructure}
We need a data structure that enables us to efficiently compute all the values and get the correct pointers at all times.
This includes, for every cluster $C$, access to its boundary edges, and its boundary size. 
We should also be able to compute efficiently, for any cut $S$ of volume at most $O(\frac m \phi)$ inside $C$, whether or not it is boundary sparse, that is we need $w(S, C\setminus S)$ and $w(S, V\setminus C)$.
We also need for each node to know in which cluster it is, and whether it is checked or unchecked.

For that, we maintain, for each node, a pointer to the cluster it is currently in, as well as whether it is checked or unchecked.
We also maintain a pointer to each of the edges it is incident to.

For each edge, we maintain whether or not it is an intercluster edge.

For each cluster $C$, we maintain a list of the nodes that constitute this cluster, whether or not it is frozen, the boundary size of the cluster and a list of all edges on its boundary.
We also maintain a copy of its mirror cluster $G/(G\setminus C)$. In \Cref{sec:mirrorcuts}, we discuss how to maintain the mirror cuts of a mirror cluster.

We call the \emph{union graph between two updates} $t_1$ and $t_2$ the graph that is the union of all edges that appears at least once between updates $t_1$ and $t_2$, that is, an edge that either exists already before update $t_1$, or is inserted at any point between times $t_1$ and $t_2$.

One main argument that is useful throughout is the fact that if we maintain a data structure on a decomposition of the graph, such that the only allowed operations are to rebuild from scratch the whole decomposition or to split a cluster such that splitting a cluster only takes time proportional to the smaller side of the cut, then the data structure can be maintained efficiently:

\begin{lemma}\label{lem:datastructure}
    Let $\C$ be a (dynamic) partition of $G$, where $G$ has $U$ updates and initially contains $m$ edges.
    Suppose that $\C$ starts as $\C=\{V\}$, and can be subject to splits during those updates, that is, an element $C$ in $\C$ can be replaced by $C\setminus S$ and $S$, for some $S\subsetneq C$.
    Let $\mathcal{D}$ be a data structure on $G$ and $\C$ such that it takes time $O(\min\{\vol(S), \vol(C\setminus S)\})$ to handle a split.

    Then over the $U$ updates, maintaining $\mathcal{D}$ takes $\tilde{O} ((m+U)T)$ time overall.
\end{lemma}

\begin{proof}
    Let us consider the union graph over the $U$ updates. 
    This graph has $m+U$ edges at most.
    For any set $X$, let $\vol^*(X)$ be the volume of $X$ in the union graph.
    Then, at any point in time, $\vol(X) \le \vol^*(X)$.

    Let us now consider a split of $C$ into $S$ and $C\setminus S$, and assume it happens after update $t$. Assume that, w.l.o.g.,  at time $t$, we have that $\vol(S) \le \vol(C\setminus S)$.
    Then, at time $t$, we have that $\vol(S)\le \vol^*(S)$ and $\vol(C\setminus S) \le \vol^*(C\setminus S)$.
    In particular, $\vol(S)\le \vol^*(C\setminus S)$ and thus $\vol(S) \le \min\{\vol^*(S),\vol^*(C\setminus S)\}$.

    Hence, if we look at the union graph, each split costs at most the side of the smallest volume. 
    Consider the half-edges (also called stubs), where each half-edge of an edge is assigned to one of the incident vertices.
    Then the volume of a node is the number of half-edges assigned to it, and the volume of a subset is the number of half-edges assigned to one of its vertices.
    Since each half-edge can only be on the smallest side of the split $\log n$ times, this concludes the proof.
\end{proof}

We thus ensure that the number of updates $U$ is $U=\Theta(\frac m {n^{o(1)}})$ as we rebuild every $\Theta(\frac {m\phi} \rho)$ with $\phi=2^{-\Theta(\log^{\frac 3  4} n )}$ and $\rho = 2^\Theta( \log^{\frac 1 2} n )$.
This allows us to amortize the time spent maintaining $\mathcal D$:

\begin{lemma}\label{lem:structuretime}
    We can maintain the data structure in $O(n^{o(1)})$ amortized time.
\end{lemma}

\begin{proof}
    Since we rebuild from scratch every time the expander decomposition is restarted, that is, every $\Theta(\frac {m\phi} \rho)$ updates, we only need to show that between every rebuild, we spend $O(mn^{o(1)})$ to maintain the data structure.

    The nontrivial part to maintain is the cluster decomposition. 
    Note that one cluster can be split into 2 clusters in two cases: if the expander decomposition outputs new intercluster edges, or if the Find and Cut Subroutine (\Cref{alg:subroutine}) outputs nodes that create a $(1-\epsilon)$-boundary sparse cut.

    Either way, the procedure will be similar: create a new cluster, traverse the smaller side of the cut, and update the status of each edge and node encountered.
    Also, at the same time, we can update the boundary edges of the original cluster.
    By moreover remembering how many intercluster edges get changed, we can also update the boundary size of the original cluster.
    We must also create a mirror cluster of the smaller side of the cut which we can do in time proportional to the volume of that cut, and update the mirror cluster of the larger cut, which involves vertex and edge deletions as well as inserting the correct number of edges to the vertex that represents the outside of the cluster. 
    This as well can be done in time proportional to the volume of the smaller side.

    Note that since we do not know beforehand which side of the cut is smaller, on can start by traversing both sides in lockstep in BFS fashion without modifying anything, to learn which side is smaller.
    We stop when we have explored completely one side.

    Hence, splitting a cluster into two takes time linear in the volume of the smaller side. We conclude by \Cref{lem:datastructure}.
\end{proof}

\begin{lemma}
    With our data structure, checking whether a cut $S$ in cluster $C$ is $(1-\epsilon)$-boundary sparse takes $O(\vol(S))$ time.
\end{lemma}

\begin{proof}
    We need to compute three quantities: $w(S, C\setminus S)$, $w(S, V\setminus C)$ and $w(C\setminus S, V\setminus C)$.
    The first two ones can be trivially computed in $O(\vol(S))$. The third one can be computed using the value stored at $C$ as $\boundary C$: $\boundary C = w(C\setminus S,V\setminus C) + w(S, V\setminus C) $.
\end{proof}

\subsection{Finding local cuts of size at most $\lmax$.}\label{sec:mirrorcuts}

Once we have a cluster hierarchy, by \Cref{prop:clusterhierarchy}, to find an approximate minimum cut in a graph, we only have to be able to find the minimum \emph{local} cut in its cluster hierarchy, dynamically, but only if its value is in $[\lmin, \lmax]$. 

For that, we will divide this task between two algorithms.
The first algorithm -- the \emph{Buffer} algorithm -- will detect whether there exists a local cut of value strictly smaller than $\lmin$. 
If it is the case, it will simply report that fact. 
If it is not the case, it will forward the updates to another algorithm that will maintain a data structure that can output the minimum local cut if its value is smaller than $\lmax$. 

\subsubsection{The buffer algorithm.}

The buffer algorithm takes as input a mirror cluster, and detects if there exists a local cut of size strictly less than $\lmin$. 
It also forwards the updates to the Mirror Cuts Algorithm (\Cref{alg:maintainmirrors}) (described below), but with a certain delay, so as to ensure that at any given time, the forwarded graph, i.e., the graph processed by  the Mirror Cuts Algorithm (\Cref{alg:maintainmirrors}), satisfies that the minimum local cut is a least $\lmin$. 
We moreover require that, if the minimum local cut in the input graph is larger than $\lmin$, then the forwarded graph is equal to the input graph. 

Let us first give a high-level intuition of how this algorithm works, before formally defining it in \Cref{alg:buffer}.
The buffer algorithm uses the following observation: If a local cut $S$ becomes smaller than $\lmin$ after an update, but it wasn't smaller before the update, then the update must necessarily have been incident to a vertex in  $S$.

Let us first explain how the algorithm works if the minimum local cut never drops below $\lmin$, and let's assume that the only requirement is for the algorithm to certify that fact. 
Then the \emph{preprocessing} (defined formally below) would be to run LocalKCut from each vertex, to ensure that no vertex is included in a small local cut. Handling updates would be as easy as running LocalKCut from each of the endpoints of every  updated edge.

What happens now if a local cut drops below $\lmin$? 
Then the algorithm finds it either in the preprocessing, or when handling the update.
It would then store that cut, and maintain its size as updates arrive, and also store all of the subsequent updates, until enough edges are inserted to bring the minimum local cut back up to $\lmin$.
In the meantime, all nodes incident to an edge update are \emph{marked} for inspection, as we refrain from running LocalKCut before the cut size of the stored cut increases again to $\lmin$.
We do run LocalKCut on all the marked nodes when the stored cut found reaches $\lmin$. 
Running it on a node $v$  might detect another cut of value less than $\lmin$, in which case we store that cut again, and start storing updates and mark nodes as discussed above.
If no such cut is found, we \emph{unmark} the node $v$. Note that this is different from \emph{checking} or \emph{unchecking} a node introduced in \Cref{sec:static}, as we are only looking for cuts of size smaller than $\lmin$, and are not checking for boundary sparseness.

Whenever there is no stored cut of size smaller than $\lmin$, and no marked (defined below) nodes, we can certify that no local cut has size less than $\lmin$, and thus we can flush all the updates to the Mirror Cuts Algorithm (\Cref{alg:maintainmirrors}), thereby ensuring that we forward the edge insertions before the edge deletions.

We maintain further two variables: $H$ and buffer.
$H$ is the graph with all updates that have been forwarded already to the Mirror Cuts Algorithm (\Cref{alg:maintainmirrors}).
\var{Buffer} is the set of all updates that have not yet been transferred to the Mirror Cuts Algorithm (\Cref{alg:maintainmirrors}).

The \emph{working graph} is the graph that we denote by $H\union$\var{Buffer}: that is the graph where all updates, also the buffered ones, are considered, or in other words, the input graph.
We consider three types of updates: edge insertions and deletions, which are typical, and \emph{set deletions}. 
For a set deletion, we are given a set of vertices of the graph, we then have to delete all the vertices and incident edges to those vertices. This is particularly useful when the cluster $C$ that we maintain our data structure on gets split into a small subgraph $S$ and a large subgraph $C\setminus S$. We can then reuse the data structure of $C$ for $C\setminus S$, by issuing a set deletion $S$ on that data structure.
We define the \emph{volumes} of updates as follows: an edge update has volume 1, while a set deletion $S$ has volume $\vol (S)$.

\begin{algo}[Buffer Algorithm]\label{alg:buffer}

    \begin{itemize}[noitemsep]
    \item Input: a mirror cluster and a sequence of updates (edge updates or set deletions) to the mirror cluster.
    \item Outputs: \begin{enumerate}
        \item \var{HasSmallCut}: a boolean that is \var{True} if and only if we found a cut of value strictly less than $\lmin$.
        \item A graph $G'$ that has no local cut of size strictly smaller than $\lmin$. If \var{HasSmallCut} is \var{False}, the graph $G'$ is the same as the input graph. If \var{HasSmallCut} is \var{True}, this output is not mandatory.
    \end{enumerate} 
        \item Maintained variables:
    \begin{enumerate}[noitemsep]
        \item \var{Buffer}: a set storing all updates that have not yet been applied to $G'$.
        \item $H$: the graph with all updates that have been applied already to $G'$.
        \item \var{HasSmallCut}: a boolean that is \var{True} if and only if we found a cut of value strictly less than $\lmin$.
        \item \var{SmallCut}: If \var{HasSmallCut} is \var{True}, this variable stores the corresponding cut.
        \item \var{SmallCutValue}: The value of the cut stored in \var{SmallCut} in the graph $H\union$\var{Buffer}.
        \item \var{Marked}: this array stores for each vertex $v$ if it is marked for processing or not.
    \end{enumerate}

    \item Preprocessing:

    $H$ begins empty, and \var{Buffer} starts with the input graph (viewed as a sequence of edge insertions). 
    \var{HasSmallCut} is \var{False}, \var{SmallCut} and \var{SmallCutValue} are $\varnothing$ and $0$ respectively. 
    All vertices are \emph{marked}, that is, \var{Marked} contains only \var{True} values.
    Then, run the \var{BatchUpdate} procedure below.

    \item Handling an update: 
    \begin{enumerate}[noitemsep]
        \item Add the update to \var{Buffer}.
        \item  If \var{HasSmallCut} is \var{True}, update \var{SmallCutValue} by decreasing (resp increasing) it by 1 for every deleted (resp inserted) edge incident to it. If \var{SmallCutValue} is now larger than $\lmin$, set \var{HasSmallCut} to \var{False}. 
        \item  If the update is an edge deletion, mark its endpoints.
        \item If the update is a set deletion, mark the nodes incident to the edges in the set.
        \item  Run the \var{BatchUpdate} procedure.
    \end{enumerate}

    \item The \var{BatchUpdate} procedure:
    
    While there is at least one marked vertex and \var{HasSmallCut} is \var{False}:
    \begin{enumerate}[noitemsep]
        \item Choose a marked vertex $v$.
        \item Run LocalKCut $10\log n \PAR{4\frac \lmax \phi}^2$ times with parameters $H\union$\var{Buffer}, $v, 4\frac \lmax \phi, \lmin$.
        \item For every cut output by LocalKCut, check if its value is strictly smaller than $\lmin$ in $H\union$\var{Buffer}. If it is the case, set \var{HasSmallCut} to \var{True}, store the corresponding \var{SmallCut} and \var{SmallCutValue}, and exit the while loop.
        \item Unmark $v$.
    \end{enumerate}
    If the while loop ends with \var{HasSmallCut} set to \var{False}, apply all the updates in \var{Buffer} to $H$and $G'$, and then set \var{Buffer} to $\varnothing$.
        \end{itemize}
\end{algo}

\subsubsection{Correctness}
As we show below, Line 2. of the while loop ensures that, with high probability, we find the local cut containing $v$ with minimal value. Since \var{SmallCutValue} always stores a true cut value as long as \var{HasSmallCut} is \var{True}, we are guaranteed that, with high probability, \var{HasSmallCut} is \var{False} as soon as the minimum local cut of the graph is above $\lmin$. Thus the only question is whether this algorithm is always able to detect the presence of a local cut of value at most $\lmin -1$.

\begin{lemma}\label{lem:bufferlemma}
    If there exists a local cut of value $c < \lmin$ after an update in the input graph $G$ to \Cref{alg:buffer}, then \var{HasSmallCut} is \var{True}, with probability of failing less than $n^{-6}$ after $n^2$ updates.  
\end{lemma}

To prove that lemma, we will use the following invariant:

\begin{invariant}\label{inv:buffer}
    In the input graph $G$ to \Cref{alg:buffer}, every cut of value at most $\lmin -1$ contains a marked vertex.
\end{invariant}

\begin{lemma}\label{lem:5.18}
\Cref{inv:buffer} holds throughout the execution of \Cref{alg:buffer}, with probability of failing less than $n^{-6}$ after a sequence of edge updates and set deletions of total volume at most $n^3$.    
\end{lemma}

\begin{proof}
This invariant is obviously true in the beginning, as every vertex is marked. The while loop in $\var{BatchUpdate}$ does not perturb this invariant either as a vertex stops being marked after running enough LocalKCut algorithms and checking the value of each cut output by it, which ensures, with high probability, that it has no cut of small value. Moreover, edge updates do not perturb this invariant: Adding edges trivially do not perturb it; Deleting an edge can cause a cut to go from having at least $\lmin$ cut-edges to at most $\lmin -1$, but then we ensure the invariant holds by marking both endpoints, one of which is guaranteed to be in that cut.
Finally, set deletions do not perturb the invariant: as it can cause a cut to go from having at least $\lmin$ cut-edges to at most $\lmin -1$, but then we ensure the invariant holds by marking all nodes incident to deleted edges, one of which is guaranteed to be in that cut.

Overall, we run a batch of $10 \log n \PAR{4 \frac \lmax \phi}^2$ LocalKCut no more than $n^4$ times if the updates have volume $n^3$: initially at most $n$ times, one on each vertex, and each update of volume $X$ forces at most $2X$ new runs, which in total is at most $n+2n^3\le n^4$.
\Cref{lem:chernoff} ensures that each batch fails with probability less than $n^{-10}$. 
A union bound gives the probability upper bound.
\end{proof}

\begin{proof}(Proof of \Cref{lem:bufferlemma})
    By \Cref{lem:5.18}, that local cut contains a marked vertex. Hence the while loop must have been ended before unmarking all vertices, which can only happen in step 3., where \var{HasSmallCut} is set to \var{True}.
\end{proof}

\subsubsection{Running time analysis}

\begin{lemma}
    The preprocessing of \Cref{alg:buffer} on graph $G=(V,E)$ takes $\tilde O\PAR{\card G \PAR{\frac \lmax \phi}^3}$ time.
    An update of volume $\nu$ is handled in $\tilde O\PAR{ \nu \PAR{\frac \lmax \phi}^3}$ amortized time.
\end{lemma}

\begin{proof}
    By \Cref{lem:kargertime}, we know that each run of LocalKCut takes $\tilde O \PAR{\frac \lmax \phi}$ time.

    During preprocessing, we run $O\PAR{(\frac \lmax \phi)^2}$ many LocalKCuts from every vertex, which corresponds to the running time claimed in the first statement.
    As every other operation takes constant time, this concludes the proof of the first statement.

    When handling an update of volume $\nu$, at most $O(\nu)$ many vertices are marked, each triggering at most $O\PAR{(\frac \lmax \phi)^2}$ which corresponds to the running time claimed in the second statement.
    As every other operation takes constant time, this concludes the proof.    
\end{proof}

We can also apply \Cref{lem:datastructure} to show that maintaining this data structure over a decomposition is efficient.

{\renewcommand\footnote[1]{}\buffer*}

\subsection{The Mirror Cuts Algorithm}

In this subsection, we present an algorithm that, given as input a mirror cluster with no local cut of size strictly smaller than $\lmin$, 

The graph given as input to our algorithms \Cref{alg:deletions} and the Mirror Cuts Algorithm (\Cref{alg:maintainmirrors}), have the guarantee that the value of the minimum local cut is always above $\lmin$. We now want to dynamically find the minimum local cut as long as it is below $\lmax$.
For that, we will maintain all \emph{mirror cuts}, whose definition we recall:

\mirrorcuts*

Note that the mirror cut of a node does not exist, if the cut of smallest value among all local cuts that include $v$ has a cut value larger than $\lmax +1$.

By \Cref{prop:clusterhierarchy}, if we know all the mirror cuts and their corresponding cut-sizes, finding a good approximation is as easy as finding the one with smallest cut-value among them.

Recall that we maintain a cluster hierarchy, inserting and deleting edges on every level. 
The rest of this section is thus devoted to finding and maintaining the mirror cuts of the hierarchy.
In fact, we only run this algorithm on every mirror cluster in the hierarchy.

The main idea is as follows: for preprocessing, we simply run LocalKCut from every vertex, which will ensure we find all mirror cuts w.h.p.
We store these cuts using the following data structure: 
\begin{itemize}
    \item For each mirror cut, we keep pointers to all the vertices for which it is the mirror cut. We also keep pointers to all the vertices it contains. We moreover keep the size of the cut.
    We delete the cut once it is not a mirror cut anymore, as described in \Cref{alg:deletions}.
    \item For each vertex in a cluster $C$ we keep a pointer to its mirror cut (If a node has multiple mirror cuts, we only maintain one) and its corresponding size. We also keep pointers to all the mirror cuts of other vertices that it is contained in.
\end{itemize}

\begin{algo}[Deleting a mirror cut]\label{alg:deletions}
    Whenever the mirror cut of a node $u$ changes from $S$ to $S'$, we remove $u$ from the list of nodes $S$ is the mirror cut of. 
    We check if this list becomes empty. 
    If it is the case, we delete all pointers from and to the nodes it contains, then delete the mirror cut.
\end{algo}

\begin{lemma}\label{lem:deletion}
    \Cref{alg:deletions} takes $O(\frac \lmax \phi)$ time as a mirror cuts has volume $O(\frac \lmax \phi)$.
\end{lemma}

During preprocessing we will build this data structure.
To process an edge insertion, we note that the edge insertion might increase the cut-size of the mirror cuts it is incident to.
We must therefore rerun LocalKCut on all the vertices whose stored cut has been affected by this insertion, as the stored cut might not be the mirror cut anymore.
Note that we can ensure that not many such nodes exist.
To process an edge deletion, we note that the deletion can decrease the size of a cut that was not a mirror cut before the deletion, but becomes one after this deletions. 
We can find all such cuts by running LocalKCut from the endpoints of the deleted edge.

We next give the details of this algorithm that we run for each cluster in the current cluster decomposition.
\begin{algo}[Mirror Cuts Algorithm]\label{alg:maintainmirrors}

\begin{itemize}[noitemsep]
    \item 
    Input: A mirror cluster of cluster $C$, where edges can be inserted and deleted, and sets of vertices can be deleted.
    
    \item Data Structure: For every node, we maintain its Mirror Cut and the corresponding value.
    Moreover, for every Mirror Cut, we maintain pointers from and to the vertices it contains.

    Preprocessing: Process using \Cref{alg:processing} every node.

    \item Handling an edge insertion: We update the value of all maintained cuts that contain its endpoints. 
    We also process using \Cref{alg:processing} all the nodes $u$ whose maintained mirror cut has seen its value increase because of the update. 

    \item Handling an edge deletion:
    We update the value of all maintained mirror cuts that contain its endpoints. We also process using \Cref{alg:processing} the endpoints of the deleted edge.

    \item Handling a set deletion:
    We delete the set from the data structure, and we update the value of all maintained mirror cuts that contain its adjacent nodes. We might also need to update the maintained cuts, as some nodes might have been deleted in the maintained cuts. We process using \Cref{alg:processing} the nodes adjacent to the deleted set.
    \end{itemize}
\end{algo}

\begin{algo}[Processing]\label{alg:processing}
    Input: A (Mirror) cluster $C'$ and a vertex $v$.

    Run LocalKCut with parameters $G=C', v, \nu = 2 \frac \lmax \phi + 2\lmax, k=\lmax$ a total number of $ 10 \log n \PAR{\frac \lmax \phi}^6 \lmax ^4$ times.
    For every cut $S$ output by LocalKCut during one of these runs that has  cut-size $c \le \lmax$, For every $u \in S$, if $S$ is a smaller cut than the one stored for $u$ execute the following steps:
    \begin{enumerate}[label=(\alph*), noitemsep]
        \item Build a data structure for $S$ if $S$ was no mirror cut before the update. 
        \item Store $S$ as mirror cut of $u$.
    \end{enumerate}
\end{algo}

We first show the correctness of the above algorithm and then analyze its running time.
\subsubsection{Correctness}
\begin{lemma}\label{lem:correctmaintainmirrorcuts}
    the Mirror Cuts Algorithm (\Cref{alg:maintainmirrors}) maintains the mirror cuts of the graph with high probability. It has probability of failing less than $n^{-6}$ after $n^3$ updates.  
\end{lemma}

\begin{proof}
    We will prove this by induction on the number of updates. The base case is clear, as after preprocessing, after all nodes are processed, the minimum local cut around each node is found and then stored.
    For the induction step, assume that after update $t$, for $t \ge 0$, we have that all cuts stored are mirror cuts.

   Let $u$ be a node and let $S$ be its maintained mirror cut after update $t$. 
   Note that $S$ could still be a mirror cut of $u$ at time $t+1$, in which case there is nothing to do, and no processing can change $S$ into another cut. 
   Indeed, \Cref{alg:processing} will only overwrite it if it finds a smaller local cut, which is impossible.
   Assume next that there exists a local cut $S'$ such that $u \in S'$ and $\boundary _{t+1} S' < \boundary _{t+1} S$.
   By definition, we have that $\boundary _{t} S \le \boundary _t S'$.
   Then either we have that the cut value of $S$ has increased, or the one of $S'$ has decreased.
   In the first case, an edge insertion incident to $S$ forces the processing of $u$. $S'$ is then found w.h.p. by processing $u$, as $S'$ has small volume and thus is found w.h.p. by all calls of LocalKCut combined. In the second case, an edge or set deletion incident to $S'$, we find $S'$ by processing the node incident to the deleted edge or to the set.
   Both processings will update the cut associated to $u$.

   Overall, we run a batch of $10 \log n \PAR{4 \frac \lmax \phi}^6\lmax^4$ LocalKCut no more than $n^4$ times during $n^3$ updates: initially at most $n$ times, one on each vertex, and each update forces at most $n$ new runs.
\Cref{lem:chernoff} ensures that each batch fails with probability less than $n^{-10}$. 
A union bound gives the probability upper bound.

\end{proof}

\subsubsection{Running time analysis}
\begin{lemma}\label{alg:vertexprocessing}
    Processing a vertex with \Cref{alg:processing} takes $\Tilde O \PAR{\PAR{\frac \lmax \phi}^9 \lmax ^4}$ time.
\end{lemma}

\begin{proof}
    Each time we process a node, we run $10\log n \PAR{\frac \lmax \phi}^6 \lmax ^4$ many times Local Karger.
    By \Cref{lem:kargertime}, each run takes $\Tilde O (\nu) $ time, where $\nu = O(\frac \lmax \phi)$ is the volume of the cut that we are looking for.
    Each run of LocalKCut can output at most $\nu$ many cuts, and each cut can update at most $\nu$ many vertices, which in turn could each  enforce the deletion of one maintained cut, which takes $\tilde O(\frac \lmax \phi)$ time by \Cref{lem:deletion}. Hence, for each run of LocalKCut, the deletions of maintained cuts can take time $\tilde O(\nu^2 \frac \lmax \phi)$.
    
    Moreover, each cut found by LocalKCut might need to be stored. In the worst case, each cut needs to be created in the data structure, and be linked to all its vertices, this takes time $\tilde O(\nu)$.
    Overall, the running time is $\Tilde O\PAR{\PAR{\frac \lmax \phi}^6 \lmax ^4 \nu^3}+\Tilde O\PAR{\PAR{\frac \lmax \phi}^6 \lmax ^4 \nu} = \Tilde O \PAR{\PAR{\frac \lmax \phi}^9 \lmax ^4}$.
\end{proof}

\begin{corollary}
    Preprocessing of the Mirror Cuts Algorithm (\Cref{alg:maintainmirrors}) on a cluster $C$ takes $\Tilde O \PAR{ \card C\PAR{\frac \lmax \phi}^9 \lmax ^4}$ time.
\end{corollary}

\begin{corollary}
    Handling an edge deletion in the Mirror Cuts Algorithm (\Cref{alg:maintainmirrors}) takes $\Tilde O \PAR{\PAR{\frac \lmax \phi}^9 \lmax ^4}$ time.
\end{corollary}

\begin{proof}
Let $e = (u,v)$ be a deleted edge. 
Consider a mirror cut $S$ containing node $u$, that is, there exists a node $u'$ such that $S$ is the mirror cut of $u'$, and $u \in S$.
Note that $S$ is $\gamma$-extreme with $\gamma = \frac 1 2$, as its cut-value is at most $\lmax$ (by definition of a mirror cut), and $S$ cannot contain a cut of size smaller than $\lmin > \frac \lmax {1.2} > \frac{\boundary S} 2 $.

If we were to run LocalKCut from $u$ once, the probability we would find $S$ is $\Omega\PAR{ \card S ^{-\frac 2 \gamma} \lmax ^{-\frac 2 \gamma +2}} = \Omega \PAR{ \PAR{\frac \lmax \phi} ^{-4} \lmax ^{-2}}$, by \Cref{lem:probaextreme}.
Since LocalKCut outputs at most $\lmax$ many cuts, this proves that the number of mirror cuts containing $u$ is $O \PAR{ \PAR{\frac \lmax \phi} ^{4} \lmax ^{3}}$.

Hence, deleting an edge can affect at most $O \PAR{ \PAR{\frac \lmax \phi} ^{4} \lmax ^{3}}$ many cuts whose size have to be updated. 
Finding those cuts in the data structure is trivial as one can access them from the endpoints of the deleted edge.

Processing both endpoints of the deleted edge takes $\Tilde O \PAR{\PAR{\frac \lmax \phi}^9 \lmax ^4}$ time by \Cref{alg:vertexprocessing}, hence the result follows.
\end{proof}
\begin{corollary}
    Handling an edge insertion in the Mirror Cuts Algorithm (\Cref{alg:maintainmirrors}) takes $\Tilde O \PAR{\PAR{\frac \lmax \phi}^{14} \lmax ^7}$ time.
\end{corollary}

\begin{proof}
Let $e = (u,v)$ be an inserted edge. 
Consider a mirror cut $S$ containing node $u$, that is, there exists a node $u'$ such that $S$ is the mirror cut of $u'$, and $u \in S$.
Note that $S$ is $\gamma$-extreme with $\gamma = \frac 1 2$, as its cut-value is at most $\lmax$ (by definition of a mirror cut), and $S$ cannot contain a cut of size smaller than $\lmin > \frac \lmax {1.2} > \frac{\boundary S} 2 $.

If we were to run LocalKCut from $u$ once, the probability we would find $S$ is $\Omega\PAR{ \card S ^{-\frac 2 \gamma} \lmax ^{-\frac 2 \gamma +2}} = \Omega \PAR{ \PAR{\frac \lmax \phi} ^{-4} \lmax ^{-2}}$, by \Cref{lem:probaextreme}.
Since LocalKCut outputs at most $\lmax$ many cuts, this proves that the number of mirror cuts containing $u$ is $O \PAR{ \PAR{\frac \lmax \phi} ^{4} \lmax ^{3}}$, each containting $O\PAR{\frac \lmax \phi}$ many nodes.

Hence, inserting an edge can affect at most $O \PAR{ \PAR{\frac \lmax \phi} ^{4} \lmax ^{3} \times \frac \lmax \phi}$ many nodes that all need to be processed, which takes $\Tilde O \PAR{\PAR{\frac \lmax \phi}^9 \lmax ^4}$ time each by \Cref{alg:vertexprocessing}, hence the result follows.
\end{proof}

\begin{corollary}
    Deleting a set $S$ in the Mirror Cuts Algorithm (\Cref{alg:maintainmirrors}) takes $O(\vol (S) \poly (\frac \lmax \phi))$ time.
\end{corollary}

\begin{proof}
    This is essentially similar as to deleting all the edges at once.

    Let $v$ be a deleted vertex and $e = (u,v)$ be a deleted edge. 
Consider a mirror cut $S$ containing node $u$, that is, there exists a node $u'$ such that $S$ is the mirror cut of $u'$, and $u \in S$.
Note that $S$ is $\gamma$-extreme with $\gamma = \frac 1 2$, as its cut-value is at most $\lmax$ (by definition of a mirror cut), and $S$ cannot contain a cut of size smaller than $\lmin > \frac \lmax {1.2} > \frac{\boundary S} 2 $.

If we were to run LocalKCut from $u$ once, the probability we would find $S$ is $\Omega\PAR{ \card S ^{-\frac 2 \gamma} \lmax ^{-\frac 2 \gamma +2}} = \Omega \PAR{ \PAR{\frac \lmax \phi} ^{-4} \lmax ^{-2}}$, by \Cref{lem:probaextreme}.
Since LocalKCut outputs at most $\lmax$ many cuts, this proves that the number of mirror cuts containing $u$ is $O \PAR{ \PAR{\frac \lmax \phi} ^{4} \lmax ^{3}}$.

Hence, deleting an edge can affect at most $O \PAR{ \PAR{\frac \lmax \phi} ^{4} \lmax ^{3}}$ many cuts whose size have to be updated. 
Finding those cuts in the data structure is trivial as one can access them from the endpoints of the deleted edge.

Processing the remaining endpoint of the deleted edge takes $\Tilde O \PAR{\PAR{\frac \lmax \phi}^9 \lmax ^4}$ time by \Cref{alg:vertexprocessing}, hence the result follows by summing over all the deleted edges.
\end{proof}

We now group our correctness and running time analyses into the following proposition:

\maintaincuts*

Together with \Cref{lem:datastructure}, we show that we can maintain the instances of the Buffer Algorithm and the Mirror Cuts Algorithm (\Cref{alg:buffer} and \Cref{alg:maintainmirrors})  efficiently in the case of a decomposition that only gets finer and finer over time:

\runningtimedatastructure*

\begin{proof}
    When splitting a cluster $C$ into $S\subsetneq C$ and $C\setminus S$, we can reuse the data structure for $C$ by issuing a set deletion for \Cref{alg:buffer} and (if output by \Cref{alg:buffer}) the Mirror Cuts Algorithm (\Cref{alg:maintainmirrors}), then create a data structure for $S$.
    Both these operations take $O(\vol (S) \frac \lmax \phi)$, so we make sure to choose $S$ as to be the one with smallest volume between $S$ and $C\setminus S$. 
    We conclude with \Cref{lem:datastructure}.
\end{proof}

\section{Generalizing for all Values of the Minimum Cut}\label{sec:combining}

In previous sections, we have presented a dynamic algorithm that is able to maintain a data structure that can output the minimum cut if $\lmin \le {\Tilde \lambda} \le \lmax$, and detect if $\lambda \ge \lmax$.

We will, in this section, explain how we use a logarithmic number of these algorithms to build an algorithm that is able to maintain the minimum cut for any value of the minimum cut, as long as it does not exceed $n^2$.

As the running time of \Cref{alg:dynamicestimate} is dependent on $\lmax$, we will need that $\lmax = O(\poly (\log n))$.
To ensure that, we will sparsify the graph as proposed by Karger~\cite{kargersparse}.
However, this sparsification needs to know what the minimum cut is, as the way it works is by simply sampling each edge of the graph with probability $p = \frac {54 \log n} {\epsilon ^2 \lambda}$, where $\lambda$ is the value of the minumum cut. 

In our case, we do not know what the value of the minimum cut is. 
To counteract that, we will sample the graph $\log_{1.1} (n^2)$ many times, each time with a different probability $p_i=\frac {54 \log n} {\epsilon ^2 \lambda}$, as follows:

Define $b_i = 1.1^i$ for all integers $i \in[1, \log_{1.1} n^2]$ and run one algorithm for each index $i$.
The $i$-th algorithm, i.e., the algorithm for index $i$, assumes that the value $\lambda$ of the minimum cut belongs to the range  $[b_i, b_{i+1})$ (called the \emph{range} of the $i$th-algorithm), and sparsifies with $p_i=\frac {54 \ln n} {\epsilon^2 b_i}$ to create graph $G_i$.

Running \Cref{alg:dynamicestimate} on $G_i$, we ensure that the $i$-th algorithm fullfills the following correctness requirement:
 If $b_i \le \lambda$, then the $i$-th algorithm either outputs a (correct) approximation of the size of the minimum cut in $G$, or correctly detects that $ \lambda_i > \lmax$, where $ \lambda_i$ is the minimum cut in $G_i$.
 If $b_i > \lambda$, then the $i$-th algorithm can output any number, i.e., there is no requirement to fulfill in this case.
Note that if $\lambda \not\in [b_i, b_{i+1})$, then the $i$-th algorithm does not need to output a correct approximation.

However, as long as 
$\lambda \in [b_i, (1+\epsilon) b_{i+1})$, the algorithm outputs a 
$(1+\epsilon)^2$-approximation for the following reason:
Simple algebra shows that $ \lambda_i > \lmax$ iff $\lambda > (1+\epsilon) b_{i+1}$.
Thus, if the $i$-th algorithm returns $ \lambda_i > \lmax$, then the output of the algorithm can be ignored by the master algorithm. 
If $ \lambda_i \le \lmax$ and $\lambda \ge \lambda_{i}$, then $\lambda \in [b_i, (1+\epsilon) b_{i+1})$ and the $i$-th algorithm outputs a $(1+\epsilon)^2$-approximation: More specifically,
it outputs a value, that is a $(1+\eps)$-approximation if $\lambda \in [b_i, b_{i+1})$, and
a $(1+\eps)^2$-approximation, if $\lambda \in [b_{i+1}, b_{i+2})$.
Thus, the algorithm 
with smallest $i$ that 
detects that $ \lambda_i \le \lmax$ returns a $(1+\eps)^2$-approximation and this is the one whose answer the master algorithm returns.
This section formalizes these ideas.

We first recall Karger's sparsification result
\kargersparse*

Note that if $ \frac {54 \ln n} {\epsilon^2 \lambda} \le p \le 1.1 \frac {54 \ln n} {\epsilon^2 \lambda} $, then the minimum cut in $G(p)$ has, with high probability, value in $[\frac{54\ln n} {\epsilon^2}, 1.1\frac{54\ln n} {\epsilon^2}]$, up to a $(1+\epsilon)$ multiplicative factor.
Therefore, setting $\lmin = (1-\epsilon)\frac{54\ln n} {\epsilon^2}, \lmax = 1.1(1+\epsilon) \frac{54\ln n} {\epsilon^2}$, ensures that in $G(p)$, the minimum cut has value in $[\lmin, \lmax]$ with high probability, with $\frac \lmax \lmin \le 1.2$, which allows us to use \Cref{alg:dynamicestimate}.

\begin{algo}[Global Dynamic Algorithm]\label{alg:global}
    Define $O(\log_{1.1}(n^2))$ many estimates of $\lambda$: $\lambda_0, \dots, \lambda_{\log_{1.1}(n^2)}$ where $b_i = 1.1^i$ for every $i \in [\log_{1.1}(n^2)]$.
    Define for every $i$, $p_i = \frac {54 \ln n} {\epsilon^2 b_i}$, and its corresponding $G(p_i)$.
    Run \Cref{alg:dynamicestimate} on every $G(p_i)$ with $\lmin = (1-\epsilon)\frac{54\ln n} {\epsilon^2}, \lmax = 1.1(1+\epsilon) \frac{54\ln n} {\epsilon^2}$.
    Let $j$ be the smallest $i$ such that the return value of \Cref{alg:dynamicestimate} on $G(p_i)$ is in $[\lmin, \lmax]$. 
    Return this value multiplied by $\frac 1 {p_j}$.
\end{algo}

\begin{lemma}
   For every $i$ such that $b_i \le \lambda$, if the return value of \Cref{alg:dynamicestimate} on $G(p_i)$ is smaller than $\lmax$, then it is at most a $(1+\epsilon)^2$ multiplicative factor away from the minimum cut value in $G$.
\end{lemma}

\begin{proof}
    We have that $p_i = \frac {54 \ln n} {\epsilon^2 b_i} \ge \frac {54 \ln n} {\epsilon^2 \lambda}$, and thus by \Cref{thm:kargersparsify}, the minimum cut in $G(p_i)$ is $(1+\epsilon)$ away from $p_i \lambda$, and thus at least $(1-\epsilon) p_i \lambda  \ge (1-\epsilon)\frac {\lambda 54 \ln n} {\epsilon^2 b_i}\ge \lmin$. 

    If moreover \Cref{alg:dynamicestimate} finds a cut of value at most $\lmax$, this ensures that the minimum cut of $G(p_i)$ is of size at most $\lmax$, and thus ensures that \Cref{alg:dynamicestimate} found a $(1+\epsilon)$ minimum cut of $G(p_i)$. 
    Multiplying the cut value by $\frac 1 p_i$ gives the desired result.
\end{proof}

This proves then the following theorem:

\maintheorem*
\section{The LocalKCut Algorithm}\label{sec:karger}

In this section, we will expand on the LocalKCut Algorithm, based on Karger's algorithm for mincut~\cite{DBLP:conf/soda/Karger93}, and similar to the work of Nalam and Saranurak~\cite{kragerlocal}.
We start by reminding how Karger's algorithm works.

\begin{algorithm}[h]
\caption{Karger$(G =(V,E))$}\label{alg:karger}
\begin{algorithmic}[1]
\While {$\card{V} >2$}
\State \label{line:randomchoice}Choose $e \in E$ uniformly at random
\State Contract $e$
\EndWhile
\State \Return $V_1$ \Comment{the vertices merged to get one of the final two vertices}
\end{algorithmic}
\end{algorithm}

Karger's main result is that this algorithm outputs any mincut with probability at least $\frac{1}{{n \choose 2}}$:
\begin{theorem}[Karger~\cite{DBLP:conf/soda/Karger93}]
Let $G$ be a \emph{simple} graph and $C$ be a minimum cut of $G$. 
Then \Cref{alg:karger} outputs $C$ (or $V\setminus C$) with probability at least $\frac{1}{{n \choose 2}}$.
\end{theorem}

Since we are going to extend on this result, we give its proof:

\begin{proof}
    First, notice that \Cref{alg:karger} outputs $C$ or $V\setminus C$ if and only if no edge from $E(C, V\setminus C)$ is ever chosen in \Cref{line:randomchoice}.
    Next, observe that after any number of contractions, the size of the minimum cut is at least $c:=\card C$. 
    Indeed, any cut in the contracted graph can be seen as a cut in the original graph.
    Therefore, the minimum degree of the contracted graph is at least $c$, and with $r$ vertices left, there are at least $\frac{rc}{2}$ edges left to choose from.
    Hence, choosing an edge from $E(C, V\setminus C)$ has probability at most $\frac{2}{r}$.
    Therefore, throughout the $n-2$ contractions, the probability that we never choose an edge from $E(C, V\setminus C)$ is at most:
    $$
    \left(1-\frac 2 n\right)\left(1-\frac 2 {n-1}\right)\dots \left(1-\frac 2 3\right) = \frac 1 {{n\choose 2}}
    $$
\end{proof}

Note that this algorithm can be modified slightly, if as an input, a random permutation $\pi$ of the edges is given: then one can simply contract the edges in the order they are given, skipping the edges whose endpoints have already been merged together by previous edges.

This observation allowed Nalam and Saranurak~\cite{kragerlocal} gave a local variant of this algorithm, where the contractions grow from a chosen vertex.
The new algorithm requires a cut value objective $k$ and a maximum volume $\nu$. We give here a slight variation to their algorithm:

\begin{algorithm}[h]
\caption{LocalKCut$(G =(V,E), v, \nu, k)$}\label{alg:localkarger} 
\begin{algorithmic}[1]
\State outputCuts $\leftarrow \varnothing$
\State $X\leftarrow \{v\}$
\While {$\vol(X) < \nu$:}
\State For every unlabeled edge $e$ in $E(X, V\setminus X)$, set $\pi(e)\in [0, 1]$ uniformly at random \\ \Comment{$\pi(i)$ represents priority. Highest priority for smallest label.}  
\State Find the edge $e = (u,v) \in E(X, V\setminus X)$ of smallest label $\pi(e)$.
\State $X\leftarrow X\union \{u,v\}$
\If{$w(X, V\setminus X) \le k$}
\State outputCuts $\leftarrow $ outputCuts$\union \{X\}$
\EndIf
\EndWhile
\State \Return outputCuts
\end{algorithmic}
\end{algorithm}

The goal of this algorithm is to find extreme sets that include $v$, and have volume at most $\nu$ and cut value at most $k$.
We have the following result, heavily inspired by Nalam and Saranurak's work.

\begin{theorem}
    Let $S$ be an extreme set of a simple graph $G$, of cut value $c \le k$ and volume $\vol(S) \le \nu$.
    Let $v \in S$ be a vertex of $S$.
    Then LocalKCut$(G,v,\nu, k, \pi)$ outputs $S$ with probability at least $\Omega(\frac 1 {\card S^2})$.
\end{theorem}

\begin{proof}
    Since \Cref{alg:localkarger} outputs all the extreme sets it has encoutered, it suffices to show that it encounters $S$ with probability $\Omega(\frac 1 {\card S^2})$.
    To do so, note that $S$ is a minimum cut of $G/(V\setminus S)$, that is, the graph obtained from $G$ by contracting $V\setminus S$ into a single vertex.
    This is true because $S$ is an extreme set.
    We can then use the same analysis as for Karger's algorithm in $G/(V\setminus S)$ to show that with probability $\Omega(\frac 1 {\card S^2})$, the edges contracted in \Cref{alg:karger} run on $G/(V\setminus S)$ will contract $S$.
    It is easy to see that those contracted edges will be the ones visited first by \Cref{alg:localkarger} (although not necessarily in the same order).
\end{proof}

We are going to extend the result to $\gamma$-extreme sets and show the following variation on the result of Nalam and Saranurak:

\begin{lemma}\label{lem:probaextreme}
    Let $S$ be a connected $\gamma$-extreme set of cut value $c \le k$ and volume $\vol(S) \le \nu$.
    Assume $\gamma$ is the inverse of an integer.
    Let $v \in S$ be a vertex of $S$.
    Then LocalKCut$(G,v,\nu, k, \pi)$ outputs $S$ with probability at least $\Omega(\card S ^{ -\frac 2 \gamma} c^{-\frac 2 \gamma+2}) $.
\end{lemma}

\begin{proof}
    Since LocalKCut will output any $\gamma$-extreme set it encounters, we only need to show that it will encounter $S$ with probability $\Omega(\card S ^{ -\frac 2 \gamma} c^{-\frac 2 \gamma+2}) $.
    Let us look at the graph $G/ (V \setminus S)$.
    Since $S$ is $\gamma$-extreme, the mincut in this graph is at least $\gamma c$.
    Let's imagine running \Cref{alg:karger} on this graph.
    Observe that after any number of contractions, the size of the minimum cut is at least $\gamma c$.
    Therefore, the minimum degree of the contracted graph is at least $\gamma c$, and with $r$ vertices left, there are at least $\frac {\gamma rc}{2}$ edges left.
    Hence, choosing an edge from $E(S, V\setminus S)$ has probability at most $\frac 2 {\gamma r}$, which is smaller than $1$ if $\gamma r > 2$.
    Therefore, over the first $\card{S} - \frac 2 {\gamma}$ contractions, the probability that no edge from $E(S, V\setminus S)$ is chosen is at least:

$$
\left(1-\frac{2}{\gamma \card S}\right)\left(1-\frac{2}{\gamma (\card S-1)}\right)\dots \left(1-\frac{2}{\gamma (\frac{2}{\gamma}+1)}\right)= \frac{1}{{\card S \choose \frac{2}{\gamma}}}=\Omega(\card S^{-\frac 2 \gamma})
$$
For the remaining $\frac 2 \gamma-1$ contractions, note that since $S$ is connected, there is at least a suitable edge to choose from, and thus the probability that we contract a suitable edge is at least $\frac 1 {c+1}$ for every contraction.
This leads to a total probability of at least $\Omega(\card S ^{ -\frac 2 \gamma} c^{-\frac 2 \gamma+2}) $.
Since these contracted edges are the ones that are going to be visited first by LocalKCut (not necessarily in the same order as the contractions), the result follows.
\end{proof}

Since \Cref{alg:localkarger} explores a set of volume at most $\nu$ and maintains a sorted list of outgoing edges, we have the following running time:

\begin{lemma}\label{lem:kargertime}
    The total running time of \Cref{alg:localkarger} with parameters $G,v,\nu, k$ is $\Tilde O (\nu)$.
\end{lemma}

Note that sometimes the output can include multiple sets whose sum of cardinalities might add up to more than $\nu$.
But since these sets are included one into the others, one can encode them in $O(\nu)$ space, not hurting the running time.

\begin{lemma}\label{lem:chernoff}
    Let $\mathcal{A}$ be a randomized algorithm and $x$ a potential output of $\mathcal{A}$. 
    Let $p$ be a lower bound on the probability that $\mathcal{A}$ outputs $x$. Let $c>0$ be a constant.

    If we run $\ell \ge \frac 1 p c \ln n$ many independent runs of $\mathcal{A}$, then the probability that $\mathcal{A}$ outputs $x$ in none of the runs is at most $n^{-c}$.
\end{lemma}

\begin{proof}
We will uses that $\ln(1-y) < -y$ for every $y \in (-\infty, 1)$.
    The probability that $\mathcal{A}$ never outputs $x$ is:

    $$
    (1-p)^\ell = \exp\PAR{\ln(1-p)\ell} \le \exp\PAR{-p\ell} \le \exp(-c\ln n) \le n^{-c} 
    $$
\end{proof}
\section{Acknowledgements}

JL would like to thank Thatchaphol Saranurak and Yaowei Long for valuable discussions on LocalKCut and how to use it for dynamic algorithms.

This project has received funding from the European Research Council (ERC) under the European Union's Horizon 2020 research and innovation programme (MoDynStruct, No. 101019564)  \includegraphics[width=0.9cm]{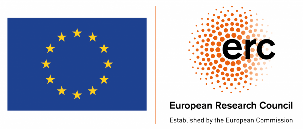} and the Austrian Science Fund (FWF) grant  \href{https://www.doi.org/10.55776/Z422}{DOI 10.55776/Z422}, grant  \href{https://www.doi.org/10.55776/I5982}{DOI 10.55776/I5982}, and grant  \href{https://www.doi.org/10.55776/P33775}{DOI 10.55776/P33775} with additional funding from the netidee SCIENCE Stiftung, 2020–2024.

\bibliographystyle{plain}
\bibliography{mincutbib}

\end{document}